\documentclass[twocolumn]{article}
\usepackage[dvipdfmx]{graphicx,xcolor}
\usepackage[fleqn]{amsmath}
\usepackage{amssymb}
\usepackage{amsthm}
\usepackage{overpic}
\usepackage{newtxtext}
\usepackage[varg]{newtxmath}
\usepackage[dvipdfm]{geometry}
\allowdisplaybreaks

\title{
  Non-Asymptotic Bounds and a General Formula for\\
  the Rate-Distortion Region of the Successive Refinement Problem\thanks{Portions of this paper were presented at the 38th Symposium on Information Theory and Its Applications \cite{matsuta2015nab}, and at the 2016 IEICE Society Conference \cite{matsuta2016general}.}
}
\author{
  Tetsunao Matsuta\thanks{tetsu@ict.e.titech.ac.jp} \and Tomohiko Uyematsu\thanks{uematsu@ict.e.titech.ac.jp}  
}
\date{\empty}

\usepackage{amsmath}
\usepackage{amssymb}
\usepackage{amsthm}
\theoremstyle{definition}
\newtheorem{thm}{Theorem}
\newtheorem{lem}{Lemma}
\newtheorem{cor}{Corollary}
\newtheorem{defn}{Definition}
\newtheorem{rem}{Remark}
\newcommand{\ra}{\rightarrow}

\newcommand{\teq}{\triangleq}
\newcommand{\A}{\alpha}

\newcommand{\cA}{\mathcal{A}}
\newcommand{\cB}{\mathcal{B}}
\newcommand{\cC}{\mathcal{C}}
\newcommand{\cD}{\mathcal{D}}
\newcommand{\cE}{\mathcal{E}}
\newcommand{\cF}{\mathcal{F}}

\newcommand{\cI}{\mathcal{I}}

\newcommand{\cM}{\mathcal{M}}

\newcommand{\cP}{\mathcal{P}}

\newcommand{\cR}{\mathcal{R}}
\newcommand{\cS}{\mathcal{S}}

\newcommand{\cU}{\mathcal{U}}

\newcommand{\cX}{\mathcal{X}}
\newcommand{\cY}{\mathcal{Y}}
\newcommand{\cZ}{\mathcal{Z}}

\newcommand{\plimsup}{\mathrm{p}\mathchar`-\!\mathop{\limsup}\limits}

\newcommand{\bfP}{\mathbf{P}}

\newcommand{\bfS}{\mathbf{S}}

\newcommand{\bfU}{\mathbf{U}}

\newcommand{\bfX}{\mathbf{X}}
\newcommand{\bfY}{\mathbf{Y}}
\newcommand{\bfZ}{\mathbf{Z}}
\def\eqo#1{\overset{\mathrm{#1}}{=}}
\def\geqo#1{\overset{\mathrm{#1}}{\geq}}
\def\leqo#1{\overset{\mathrm{#1}}{\leq}}

\begin{document}
\maketitle

\renewcommand{\thefootnote}{\fnsymbol{footnote}}
\footnote[0]{The authors are with Dept.\ of Information and Communications Engineering, Tokyo Institute of Technology, Tokyo, 152-8552 Japan.}
\renewcommand{\thefootnote}{\arabic{footnote}}

{\small
  \textbf{SUMMARY}  
  In the successive refinement problem, a fixed-length sequence emitted from an information source is encoded into two codewords by two encoders in order to give two reconstructions of the sequence. One of two reconstructions is obtained by one of two codewords, and the other reconstruction is obtained by all two codewords. For this coding problem, we give non-asymptotic inner and outer bounds on pairs of numbers of codewords of two encoders such that each probability that a distortion exceeds a given distortion level is less than a given probability level. We also give a general formula for the rate-distortion region for general sources, where the rate-distortion region is the set of rate pairs of two encoders such that each maximum value of possible distortions is less than a given distortion level.

}
\section{Introduction}
The successive refinement problem is a fixed-length lossy source coding problem with many terminals (see Fig.\ \ref{fig: successive refinement problem}). In this coding problem, a fixed-length sequence emitted from an information source is encoded into two codewords by two encoders in order to give two reconstructions of the sequence. One of two reconstructions is obtained by one of two codewords by using a decoder, and the other reconstruction is obtained by all two codewords by using the other decoder.

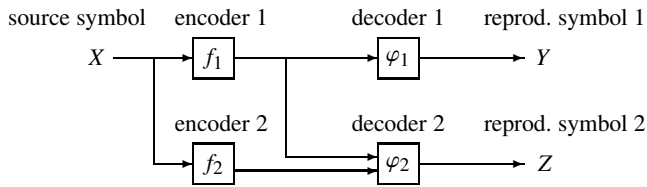
\begin{figure}[t]
    \centering
\begin{picture}(240, 70)
    \small
 \put(-10, 60){source symbol}
 \put(20, 45){$X$}
 \put(30, 47.5){\vector(1, 0){30}}
 \put(45, 47.5){\line(0, -1){40}}
 \put(45, 7.5){\vector(1, 0){15}}

 \put(53, 60){encoder 1}
 \put(60, 40){\framebox(15, 15){$f_{1}$}} 
 \put(75, 47.5){\vector(1, 0){55}}
 \put(95, 47.5){\line(0, -1){37.5}}
 \put(95, 10){\vector(1, 0){35}}

  \put(53, 20){encoder 2}
  \put(60, 0){\framebox(15, 15){$f_{2}$}} 
  \put(75, 5){\vector(1, 0){55}}

 \put(120, 60){decoder 1}
 \put(130, 40){\framebox(15, 15){$\varphi_{1}$}}
 \put(145, 47.5){\vector(1, 0){40}}

 \put(120, 20){decoder 2}
  \put(130, 0){\framebox(15, 15){$\varphi_{2}$}}
 \put(145, 7.5){\vector(1, 0){40}}
 
 \put(170, 60){reprod.\ symbol 1}
 \put(190, 45){$Y$}

  \put(170, 20){reprod.\ symbol 2}
 \put(190, 5){$Z$} 
\end{picture}
    \label{fig: successive refinement problem}
    \caption{Successive refinement problem}
\end{figure}

An important parameter of the successive refinement problem is a pair of \textit{rates} of two encoders such that each distortion between the source sequence and a reconstruction is less than a given distortion level. The set of these pairs when the length (blocklength) of the source sequence is unlimited is called the \textit{rate-distortion region}. Since a codeword is used in both decoders, we cannot always optimize rates like the case where each codeword is used for each reconstruction separately. However, there are some cases where we can achieve the optimum rates. Necessary and sufficient conditions for such cases were independently given by Koshelev \cite{koshelev1981estimation}, \cite{koshelev1980hierarchical} and Equitz and Cover \cite{equitz1991successive}. The complete characterization of the rate-distortion region for discrete stationary memoryless sources was given by Rimoldi \cite{272493}. Yamamoto \cite{490549} also gave the rate-distortion region as a special case of a more general coding problem. Later, Effros \cite{782109} characterized the rate-distortion region for discrete stationary ergodic and non-ergodic sources.

Recently, the asymptotic analysis of the second-order rates to the blocklength becomes an active target of the study. Especially, for the successive refinement problem, No et al.\ \cite{7445859} and Zhou et al.\ \cite{zhou2017second} gave a lot of results to the set of second-order rates for discrete and Gaussian stationary memoryless sources. No et al.\ \cite{7445859} considered \textit{separate} excess-distortion criteria such that a probability that a distortion exceeds a given distortion level is less than a given probability level separately for each reconstruction. On the other hand, Zhou et al.\ \cite{zhou2017second} considered the \textit{joint} excess-distortion criterion such that a probability that either of distortions exceeds a given distortion level is less than a given probability level. Although they also gave several non-asymptotic bounds on the set of pairs of rates, they mainly focus on the asymptotic behavior of the set.

On the other hand, in this paper, we consider non-asymptotic bounds on pairs of rates in finite blocklengths. Especially, since a rate is easily calculated by a number of codewords, we focus on pairs of two numbers of codewords. Although we adopt separate excess-distortion criteria, our result can be easily applied to the joint excess-distortion criterion. We give inner and outer bounds on pairs of numbers of codewords. These bounds are characterized by using the \textit{smooth max R\'enyi divergence} introduced by Warsi \cite{6691300}. For the point-to-point lossy source coding problem, we also used the smooth max R\'enyi divergence to characterize the rate-distortion function which is the minimum rate when the blocklength is unlimited \cite{uyematsu2014ITWrevisiting}. Proof techniques are similar to that of \cite{uyematsu2014ITWrevisiting}, but we employ several extended results for the successive refinement problem. The inner bound is derived by using an extended version of the previous lemma \cite[Lemma~2]{uyematsu2014ITWrevisiting}. We give this lemma as a special case of an extended version of the previous generalized covering lemma \cite[Lemma~1]{matsuta2016new}. The outer bound is derived by using an extended version of the previous converse bound \cite[Lemma~4]{uyematsu2014ITWrevisiting}.

In this paper, we also consider the rate-distortion region for general sources. In this case, we adopt the maximum-distortion criterion such that the maximum value of possible distortion is less than a given distortion level for each reconstruction. By using the spectral sup-mutual information rate (cf.\ \cite{hanspringerinformation}) and the non-asymptotic inner and outer bounds, we give a general formula for the rate-distortion region. We show that our rate-distortion region coincides with the region obtained by Rimoldi \cite{272493} when a source is discrete stationary memoryless. Furthermore, we consider a mixed source which is a mixture of two sources and show that the rate-distortion region is the intersection of those of two sources.

The rest of this paper is organized as follows. In Section \ref{sec: preliminaries}, we provide some notations and the formal definition of the successive refinement problem. In Section \ref{sec: covering lemma}, we give several lemmas for an inner bound on pairs of numbers of codewords and the rate-distortion region. These lemmas are extended versions of our previous results \cite[Lemma 2]{uyematsu2014ITWrevisiting} and \cite[Lemma 1]{matsuta2016new}. In Section \ref{sec: inner and outer bounds}, we give outer and inner bounds using the smooth max R\'enyi divergence on pairs of numbers of codewords. In Section \ref{sec: the limit of the rate region}, we give a general formula for the rate-distortion region. In this section, we consider the rate-distortion region for discrete stationary memoryless sources and mixed sources. In Section \ref{sec: conclusion}, we conclude the paper.

\section{Preliminaries}
\label{sec: preliminaries}
Let $\mathbb{N}$, $\mathbb{R}$, and $\mathbb{R}_{\geq 0}$ be sets of positive integers, real numbers, and non-negative real numbers, respectively.

Unless otherwise stated, we use the following notations. For a pair of integers $i \leq j$, the set of integers $\{i, i+1, \cdots, j\}$ is denoted by $[i: j]$. For finite or countably infinite sets $\cX$ and $\cY$, the set of all probability distributions over $\cX$ and $\cX \times \cY$ are denoted by $\cP_{\cX}$ and $\cP_{\cX \cY}$, respectively. The set of all conditional probability distributions over $\cX$ given $\cY$ is denoted by $\cP_{\cX|\cY}$. The probability distribution of a random variable (RV) $X$ is denoted by the subscript notation $P_{X}$, and the conditional probability distribution for $X$ given an RV $Y$ is denoted by $P_{X|Y}$. The $n$-fold Cartesian product of a set $\cX$ is denoted by $\cX^n$ while an $n$-length sequence of symbols $(a_1,a_2,\cdots,a_n)$ is denoted by $a^n$. The sequence of RVs $\{X^n\}_{n=1}^\infty$ is denoted by the bold-face letter $\bfX$. Sequences of probability distributions $\{P_{X^n}\}_{n = 1}^{\infty}$ and conditional probability distributions $\{P_{X^n|Y^n}\}_{n = 1}^{\infty}$ are denoted by bold-face letters $\bfP_{\bfX}$ and $\bfP_{\bfX|\bfY}$, respectively.


For the successive refinement problem (Fig.~\ref{fig: successive refinement problem}), let $\cX$, $\cY$, and $\cZ$ be finite or countably infinite sets, where $\cX$ represents the source alphabet, and $\cY$ and $\cZ$ represent two reconstruction alphabets. Let $X$ over $\cX$ be an RV which represents a single source symbol. Since the source can be characterized by $X$, we also refer to it as the source. When we consider $\cX$ as an $n$-fold Cartesian product of a certain finite or countably infinite set, we can regard the source symbol $X$ as an $n$-length source sequence. Thus, for the sake of brevity, we deal with the single source symbol unless otherwise stated.

Two encoders encoder 1 and encoder 2 are represented as functions $f_1:\cX\ra[1:M_1]$ and $f_2:\cX\ra[1:M_2]$, respectively, where $M_1$ and $M_2$ are positive integers which denote numbers of codewords. Two decoders decoder 1 and decoder 2 are represented as functions $\varphi_1:[1:M_1]\ra\cY$ and $\varphi_2:[1:M_1]\times [1:M_2]\ra\cZ$, respectively. We refer to a tuple of encoders and decoders $(f_1,f_2,\varphi_1,\varphi_2)$ as a \textit{code}. In order to measure distortions between the source symbol and reconstruction symbols, we introduce distortion measures defined by functions $d_1: \cX \times \cY \ra [0, +\infty)$ and $d_2: \cX \times \cZ \ra [0, +\infty)$.

We define two events of exceeding given distortion levels $D_1 \geq 0$ and $D_2 \geq 0$ as follows:
\begin{align*}
    \cE_1(D_1) &\teq \{d_1(X,\varphi_1(f_1(X))) > D_1\},\\
    \cE_2(D_2) &\teq \{d_2(X,\varphi_2(f_1(X),f_2(X))) > D_2\}.
\end{align*}
Then, we define the achievability under the excess-distortion criterion.
\begin{defn}
    For positive integers $M_1,M_2$, real numbers $D_1, D_2 \geq 0$, and $\epsilon_1,\epsilon_2 \in [0, 1]$, let $M = (M_1,M_2)$, $D = (D_1,D_2)$, and $\epsilon =(\epsilon_1,\epsilon_2)$. Then, for a source $X$, we say $(M,D)$ is \textit{$\epsilon$-achievable} if and only if there exists a code $(f_1,f_2,\varphi_1,\varphi_2)$ such that numbers of codewords of encoder 1 and encoder 2 are $M_1$ and $M_2$, respectively, and
    \begin{align*}
        \Pr\{\cE_i(D_i)\}\leq \epsilon_i, \quad \forall i\in\{1,2\}.
    \end{align*}
\end{defn}
In what follows, for constants $M_1,M_2, D_1, D_2,\epsilon_1,$ and $\epsilon_2$, we often use the above simple notations: $M = (M_1,M_2)$, $D = (D_1,D_2)$, and $\epsilon =(\epsilon_1,\epsilon_2)$. In this setting, we consider the set of all pairs $(M_1,M_2)$ of numbers of codewords under the excess-distortion criterion. According to the $\epsilon$-achievability, this set is defined as follows:
\begin{defn}
    For a source $X$, real numbers $D_1, D_2 \geq 0$, and $\epsilon_1,\epsilon_2 \in [0, 1]$, we define
    \mathindent=0mm
    \begin{align*}
        \cM(D,\epsilon| X) \teq \{(M_1,M_2) \in \mathbb{N}^{2}: (M,D) \mbox{ is $\epsilon$-achievable}\}.
    \end{align*}
    \mathindent=7mm
\end{defn}

Basically, this paper deals with a coding for a single source symbol. However, in Section \ref{sec: the limit of the rate region}, we deal with the coding for an $n$-length source sequence. Hence in that section, by abuse of notation, we regard the above sets $\cX$, $\cY$, and $\cZ$ as $n$-fold Cartesian products $\cX^n$, $\cY^n$, and $\cZ^n$, respectively. We also regard source symbol $X$ on $\cX$ as an $n$-length source sequence $X^n$ on $\cX^n$. Then we call the sequence $\bfX = \{X^n\}_{n = 1}^{\infty}$ of source sequences the \textit{general source} that is not required to satisfy the consistency condition.

We use the superscript $(n)$ for a code, distortion measures, and numbers of codewords (e.g., $(f_1^{(n)}, f_2^{(n)}, \varphi_1^{(n)}, \varphi_2^{(n)})$) to make clear that we are dealing with source sequences of length $n$. For a code, we define rates $R_{1}^{(n)}$ and $R_{2}^{(n)}$ as
\begin{align*}
    R_{i}^{(n)} \teq \frac{1}{n} \log M_{i}^{(n)}, \quad \forall i \in \{1, 2\}.
\end{align*}
Hereafter, $\log$ means the natural logarithm.

We introduce maximum distortions for a sequence of codes. To this end, we define the limit superior in probability \cite{hanspringerinformation}.
\begin{defn}[Limit superior in probability]
    For an arbitrary sequence $\bfS = \{S^n\}_{n = 1}^{\infty}$ of real-valued RVs, we define the limit superior in probability by
    \begin{align*}
        \plimsup_{n\ra\infty} S_n &\teq \inf\left\{\A : \lim_{n\ra \infty}\Pr\left\{S_n > \A\right\} = 0 \right\}.
    \end{align*}
\end{defn}
Now we introduce the maximum distortions:
\begin{align*}
    \plimsup_{n\ra\infty} d_1^{(n)}(X^n, \varphi_1^{(n)}(f_1^{(n)}(X^n))),\\
    \plimsup_{n\ra\infty} d_2^{(n)}(X^n, \varphi_2^{(n)}(f_1^{(n)}(X^n),f_2^{(n)}(X^n))).
\end{align*}
Then, we define the achievability under the maximum distortion criterion.
\begin{defn}
    For real numbers $R_1, R_2 \geq 0$, let $R = (R_1, R_2)$. Then, for a general source $\bfX$, and real numbers $D_1, D_2 \geq 0$,  we say a pair $(R, D)$ is \textit{fm-achievable} if and only if there exists a sequence $\{(f_1^{(n)}, f_2^{(n)}, \varphi_1^{(n)}, \varphi_2^{(n)})\}$ of codes satisfying
    \begin{align*}
        \plimsup_{n\ra\infty} d_1^{(n)}(X^n, \varphi_1^{(n)}(f_1^{(n)}(X^n))) &\leq D_1,\\
        \plimsup_{n\ra\infty} d_2^{(n)}(X^n, \varphi_2^{(n)}(f_1^{(n)}(X^n),f_2^{(n)}(X^n))) &\leq D_2,
    \end{align*}
    and
    \begin{align*}
        \limsup_{n\ra\infty} R_i^{(n)} \leq R_i, \quad \forall i \in \{1, 2\}.
    \end{align*}
\end{defn}
In what follows, for constants $R_1$ and $R_2$, we often use the above simple notation: $R = (R_1, R_2)$. In this setting, we consider the set of all rate pairs under the maximum distortion criterion. According to the fm-achievability, this set, usually called the rate-distortion region, is defined as follows:
\begin{defn}[Rate-distortion region]
    For a general source $\bfX$ and real numbers $D_1, D_2 \geq 0$, we define
    \begin{align*}
        \cR(D| \bfX) \teq \{(R_1, R_2) \in \mathbb{R}_{\geq 0}^{2}: (R, D) \mbox{ is fm-achievable}\}.
    \end{align*}
\end{defn}
\begin{rem}
    We can show that the rate-distortion region $\cR(D| \bfX)$ is a closed set by the definition and using the diagonal line argument (cf.\ \cite{hanspringerinformation}).
\end{rem}

We note that when we regard $X$ as $n$-length sequence in the definition of $\cM(D, \epsilon| X)$, it gives a non-asymptotic region of pairs of rates for a given finite blocklength.

\section{Covering Lemma}
\label{sec: covering lemma}
In this section, we introduce some useful lemmas and corollaries for an inner bound on the set $\cM(D, \epsilon| X)$ and $\cR(D| \bfX)$.

The next lemma is the most basic and important result in the sense that all subsequent results in this section are given by this lemma.
\begin{lem}
    \label{lem: Covering Lemma}
    Let $A \in \cA$ be an arbitrary RV, and $\tilde B \in \cB$ and $\tilde C \in \cC$ be RVs such that the pair $(\tilde B, \tilde C)$ is independent of $A$. For an integer $M_1 \geq 1$, let $\tilde B_1,\tilde B_2,\cdots, \tilde B_{M_1}$ be RVs which are independent of each other and of $A$, and each distributed according to $P_{\tilde B}$. For an integer $i\in [1 : M_1]$ and $M_2 \geq 1$, let $\tilde C_{i,1},\tilde C_{i,2},\cdots, \tilde C_{i,M_2}$ be RVs which are independent of each other and of $A$, and each distributed according to $P_{\tilde C|\tilde B}(\cdot | \tilde B_i)$. Then, for any set $\cF\subseteq \cA\times \cB\times \cC$, we have
    \begin{align}
        &\Pr\left\{ \bigcap_{i=1}^{M_1} \bigcap_{j=1}^{M_2}\{(A,\tilde B_i, \tilde C_{i,j}) \notin \cF\}\right\}\notag\\
        &=\mathrm{E} \left[\mathrm{E} \left[ \mathrm{E} \left[\mathbf{1}\{ (A,\tilde B, \tilde C) \notin \cF\} \middle| A, \tilde B \right]^{M_2} \middle| A \right]^{M_1}\right],
        \label{equ: lem: covering lemma}
    \end{align}
    where $\mathbf{1}\{\cdot\}$ denotes the indicator function, $\mathrm{E} \left[ \cdot \right]$ denotes the expectation operator, and $\mathrm{E}\left[ \cdot \right]^{M}$ denotes the $M$-th power of the expectation, i.e., $\mathrm{E}\left[ \cdot \right]^{M} = \left(\mathrm{E}\left[ \cdot \right]\right)^{M}$.
\end{lem}
\begin{proof}
    We have
    \begin{align*}
        &\Pr\left\{ \bigcap_{i=1}^{M_1} \bigcap_{j=1}^{M_2}\{(A,\tilde B_i, \tilde C_{i,j}) \notin \cF\}\right\}\\
        &= \sum_{a\in\cA}\sum_{(\tilde b_{1}\cdots,\tilde b_{M_1})\in\cB^{M_1}}
        \sum_{(\tilde c_{1,1},\cdots, \tilde c_{1,M_2},c_{2,1}\cdots, \tilde c_{M_1,M_2}) \in\cC^{M_1M_2}}\\
        &\quad\times \left(P_{A}(a)\prod_{i=1}^{M_1}P_{\tilde B}(\tilde b_i)\prod_{j=1}^{M_2}P_{\tilde C|\tilde B}(\tilde c_{i,j}|\tilde b_i)\right)\\
        &\quad \times \left(\prod_{i=1}^{M_1} \prod_{j=1}^{M_2} \mathbf{1}\{(a,\tilde b_i, \tilde c_{i,j}) \notin \cF\}\right)\\
        &= \sum_{a\in\cA}\sum_{(\tilde b_{1}\cdots,\tilde b_{M_1})\in\cB^{M_1}}
        \sum_{(\tilde c_{1,1},\cdots, \tilde c_{1,M_2},c_{2,1}\cdots, \tilde c_{M_1,M_2}) \in\cC^{M_1M_2}}\\
        &\quad\times P_{A}(a)\prod_{i=1}^{M_1}P_{\tilde B}(\tilde b_i)\prod_{j=1}^{M_2}P_{\tilde C|\tilde B}(\tilde c_{i,j}|\tilde b_i)\\
        &\quad \times \mathbf{1}\{(a,\tilde b_i, \tilde c_{i,j}) \notin \cF\}\\
        &= \sum_{a\in\cA}P_{A}(a)\prod_{i=1}^{M_1}\sum_{\tilde b_i\in\cB}P_{\tilde B}(\tilde b_i)
        \prod_{j=1}^{M_2}\sum_{\tilde c_{i,j}\in\cC}
        P_{\tilde C|\tilde B}(\tilde c_{i,j}|\tilde b_i)\\
        &\quad \times \mathbf{1}\{(a,\tilde b_i, \tilde c_{i,j}) \notin \cF\}\\
        &= \sum_{a\in\cA}P_{A}(a)\prod_{i=1}^{M_1}\sum_{\tilde b\in\cB}P_{\tilde B}(\tilde b)
        \prod_{j=1}^{M_2}\sum_{\substack{\tilde c\in\cC:\\ (a,\tilde b, \tilde c) \notin \cF}}
        P_{\tilde C|\tilde B}(\tilde c|\tilde b)\\
        &= \sum_{a\in\cA} P_{A}(a) \left(\sum_{\tilde b\in\cB}P_{\tilde B}(\tilde b) \left(\sum_{\substack{\tilde c\in\cC:\\ (a,\tilde b, \tilde c) \notin \cF}} P_{\tilde C|\tilde B}(\tilde c|\tilde b) \right)^{M_2} \right)^{M_1}.
    \end{align*}
    By recalling that $(\tilde B, \tilde C)$ is independent of $A$, this coincides with the right-hand side (RHS) of \eqref{equ: lem: covering lemma}.
\end{proof}           

This lemma implies an exact analysis of the error probability of covering a set $\cA$ in terms of a given condition $\cF$ by codewords $\{\tilde B_{i}\}$ and $\{\tilde C_{i, j}\}$ of random coding. Hence, this lemma can be regarded as an extended version of \cite[Theorem 9]{kostina2012fixed}.

Although the above lemma gives an exact analysis, it is difficult to use it for characterizing an inner bound on pairs of numbers of codewords and the rate-distortion region. Instead of it, we will use the next convenient lemma.
\begin{lem} 
    \label{lem: Bound for Covering Lemma}
    Let $A \in \cA$, $B \in \cB$, and $C \in \cC$ be arbitrary RVs, and $\tilde B \in \cB$ and $\tilde C \in \cC$ be RVs such that the pair $(\tilde B, \tilde C)$ is independent of $A$. Let $\psi_1:\cA\times \cB \ra [0,1]$ be a function and $\alpha_1 \in [0, 1]$ be a constant such that
    \begin{align}
        & P_{A}(a) P_{\tilde B}(b) \geq \alpha_1 \psi_1(a,b) P_{AB}(a, b),\notag\\
        & \quad \forall (a,b)\in\cA\times\cB.
        \label{equ: alpha_1psi_1condition}
    \end{align}
    Furthermore, let $\psi_2:\cA\times\cB\times\cC \ra [0,1]$ be a function and $\alpha_2 \in [0, 1]$ be a constant such that
    \begin{align}   
        &P_{AB}(a, b) P_{\tilde C|\tilde B}(c|b) \geq \alpha_2 \psi_2(a,b,c) P_{ABC}(a, b, c),\notag\\
        &\quad \forall (a,b,c)\in\cA\times\cB\times\cC.
        \label{equ: alpha_2psi_2condition}      
    \end{align}
    Then, for any set $\cF\subseteq \cA\times \cB\times \cC$, we have
    \begin{align}
        &\mathrm{E} \left[\mathrm{E} \left[ \mathrm{E} \left[ \mathbf{1}\{ (A,\tilde B, \tilde C) \notin \cF\} \middle| A, \tilde B \right]^{M_2} \middle| A \right]^{M_1}\right]\notag\\
        &\leq 1- \mathrm{E}[\psi_1(A,B) \psi_2(A,B,C)] + \Pr\left\{(A,B,C) \notin \cF\right\}\notag\\
        &\quad +e^{-\alpha_2 M_2-\log\alpha_1} +e^{-\alpha_1 M_1}.
        \label{equ: lem: bound for covering lemma}
    \end{align}
\end{lem}
\begin{proof}
    We have
    \mathindent=0mm
    \begin{align*}
        &\mathrm{E} \left[\mathrm{E} \left[ \mathrm{E} \left[\mathbf{1}\{ (A,\tilde B, \tilde C) \notin \cF\} \middle|A, \tilde B \right]^{M_2} \middle| A \right]^{M_1}\right]\\
        &= \sum_{\substack{a\in\cA:\\ P_{A}(a) > 0}} P_{A}(a) \left(\sum_{b\in\cB}P_{\tilde B}(b)\right.\\
        &\quad \times \left. \left(1-\sum_{\substack{c\in\cC:\\ (a,b, c) \in \cF}} P_{\tilde C|\tilde B}(c|b) \right)^{M_2} \right)^{M_1}\\
        &\overset{\rm (a)}{\leq} \sum_{\substack{a\in\cA:\\ P_{A}(a)>0}} P_{A}(a) \left(\sum_{b\in\cB}P_{\tilde B}(b)\right.\\
        &\quad \left. \times  \left(1 - \alpha_2 \sum_{\substack{c\in\cC: \\(a,b, c) \in \cF,\\ P_{B|A}(b|a)>0}} \psi_2(a,b,c) P_{C|AB}(c|a,b) \right)^{M_2} \right)^{M_1}\\
        &\overset{\rm (b)}{\leq} \sum_{\substack{a\in\cA:\\ P_{A}(a)>0}} P_{A}(a) \left(\sum_{b\in\cB}P_{\tilde B}(b) \left(1 - \sum_{\substack{c\in\cC: \\(a,b, c) \in \cF,\\ P_{B|A}(b|a)>0}} \psi_2(a,b,c) \right. \right.\\
        &\quad \left. \left. \times P_{C|AB}(c|a,b) +e^{-\alpha_2 M_2} \right) \right)^{M_1}\\
        &\overset{\rm (c)}{\leq} \sum_{\substack{a\in\cA:\\ P_{A}(a)>0}} P_{A}(a) \left(1-\sum_{\substack{b\in\cB:\\ P_{B|A}(b|a)>0}} \alpha_1\psi_1(a,b)P_{B|A}(b|a)\right. \\
        &\quad \times \left. \sum_{\substack{c\in\cC: \\ (a,b, c) \in \cF}} \psi_2(a,b,c) P_{C|AB}(c|a,b) + e^{-\alpha_2 M_2} \right)^{M_1}\\
        &= \sum_{\substack{a\in\cA:\\ P_{A}(a)>0}} P_{A}(a) \left(1-\alpha_1 \left(\sum_{\substack{(b, c)\in\cB\times\cC:\\ (a,b, c) \in \cF,\\ P_{B|A}(b|a)>0}} \psi_1(a,b)\psi_2(a,b,c) \right. \right.\\
        &\quad \times \left. \left. P_{BC|A}(b,c|a) -e^{-\alpha_2 M_2-\log\alpha_1} \right) \right)^{M_1},
    \end{align*}
    \mathindent=7mm
    where (a) comes from (\ref{equ: alpha_2psi_2condition}),
    (b) follows since $(1-xy)^M \leq 1 - y + e^{-xM}$ for $0\leq x, y \leq 1$ and $M>0$ (cf.\ \cite[Lemma 10.5.3]{cover2006eit}), and (c) comes from (\ref{equ: alpha_1psi_1condition}). Since the probability is not greater than $1$, we have
    \mathindent=0mm
    \begin{align*}
        &\mathrm{E} \left[\mathrm{E} \left[ \mathrm{E} \left[\mathbf{1}\{ (A,\tilde B, \tilde C) \notin \cF\} \middle|A, \tilde B \right]^{M_2} \middle| A \right]^{M_1}\right]\\
        &\leq \sum_{\substack{a\in\cA:\\ P_{A}(a)>0}} P_{A}(a) \left(1-\alpha_1 \left|\sum_{\substack{(b, c)\in\cB\times\cC: \\ (a,b, c) \in \cF,\\ P_{B|A}(b|a)>0}} \psi_1(a,b)\psi_2(a,b,c) \right. \right.\\
        &\quad \times \left. \left.  P_{BC|A}(b,c|a) - e^{-\alpha_2 M_2-\log\alpha_1}\right|^{+} \right)^{M_1}\\
        &\overset{\rm (a)}{\leq} \sum_{\substack{a\in\cA:\\ P_{A}(a)>0}} P_{A}(a) \left(1 - \left|\sum_{\substack{(b,c)\in\cB\times\cC: \\ (a,b, c) \in \cF, \\ P_{B|A}(b|a)>0}} \psi_1(a,b) \psi_2(a,b,c) \right. \right.\\
        &\quad \left. \left. \times P_{BC|A}(b,c|a) - e^{-\alpha_2 M_2-\log\alpha_1}\right|^+ +e^{-\alpha_1 M_1} \right)\\
        &\leq 1-\sum_{\substack{(a,b,c)\in\cA\times\cB\times\cC: \\ (a,b,c) \in \cF,P_{AB}(a,b)>0}} \psi_1(a,b) \psi_2(a,b,c) P_{ABC}(a,b,c)\\
        & \quad +e^{-\alpha_2 M_2-\log\alpha_1} + e^{-\alpha_1 M_1}\\
        &\leqo{(b)} 1- \mathrm{E}[\psi_1(A,B) \psi_2(A,B,C)] + \Pr\left\{(A,B,C) \notin \cF\right\}\\
        &\quad +e^{-\alpha_2 M_2-\log\alpha_1} + e^{-\alpha_1 M_1},
    \end{align*}
    \mathindent=7mm
    where $|x|^+ = \max\{0,x\}$, (a) follows since $(1-xy)^M \leq 1 - y + e^{-xM}$ for $0\leq x, y \leq 1$ and $M>0$, and (b) comes from the fact that $\psi_1(a, b) \psi_2(a, b, c) \leq 1$.
\end{proof}

The importance of this lemma is to be able to change RVs from $(A, \tilde B, \tilde C)$ to arbitrary correlated RVs $(A, B, C)$. This makes it possible to characterize an inner bound on pairs of numbers of codewords and the rate-distortion region.

Lemma \ref{lem: Bound for Covering Lemma} can be regarded as an extended version of our previous lemma \cite[Lemma 1]{matsuta2016new} to multiple correlated RVs. Hence, like the previous lemma, by changing functions and constants, it gives many types of bounds such as the following two corollaries.
\begin{cor}
    For any set $\cF\subseteq \cA\times \cB\times \cC$, any real numbers $\gamma_1,\gamma_2 \in \mathbb{R}$, and any integers $M_1,M_2 \geq 1$ such that $M_1\geq \exp(\gamma_1)$ and $M_2\geq \exp(\gamma_2)$, we have
    \begin{align*}
        &\mathrm{E}\left[\mathrm{E}\left[ \mathrm{E}\left[\mathbf{1}\{ (A,\tilde B, \tilde C) \notin \cF\} \middle|A, \tilde B \right]^{M_2} \middle| A \right]^{M_1}\right]\\
        &\leq \Pr\left\{ \log\frac{P_{B|A}(B|A)}{P_{\tilde B}(B)} > \log M_1 - \gamma_1\right.\\
        &\quad \left. \mbox{ or } \log\frac{P_{C|AB}(C|A,B)}{P_{\tilde C|\tilde B}(C|B)} > \log M_2 - \gamma_2\right\}\\
        &\quad + \Pr\left\{(A,B,C) \notin \cF\right\} + e^{-\exp(\gamma_2)-\gamma_1+\log M_1} +e^{-\exp(\gamma_1)}.
    \end{align*}
\end{cor}
\begin{proof}
    Let $\alpha_1 = \frac{\exp(\gamma_1)}{M_1}$, $\alpha_2 = \frac{\exp(\gamma_2)}{M_2}$,
    \begin{align*}
        \psi_1(a,b) &= \mathbf{1}\left\{\log\frac{P_{B|A}(b|a)}{P_{\tilde B}(b)} \leq \log M_1 - \gamma_1 \right\},\\
        \psi_2(a,b,c) &= \mathbf{1}\left\{\log\frac{P_{C|AB}(c|a,b)}{P_{\tilde C|\tilde B}(c|b)} \leq \log M_2 - \gamma_2\right\}.
    \end{align*}
    Then, we can easily check that these constants and functions satisfy \eqref{equ: alpha_1psi_1condition} and \eqref{equ: alpha_2psi_2condition}. Plugging these functions and constants into \eqref{equ: lem: bound for covering lemma}, we have the desired bound.
\end{proof}

This corollary can be regarded as a bound in terms of the \textit{information spectrum} (cf.\ \cite{hanspringerinformation}). To the best of our knowledge, this type of bound has not been reported so far (although, there are some converse bounds \cite[Lemma 15]{zhou2017second} and \cite[Theorem 3]{kostina2017rate-distortion}).

On the other hand, the next corollary gives a bound in terms of the smooth max R\'enyi divergence $D_{\infty}^{\delta}(P\| Q)$ defined as
\begin{align*}
    D_{\infty}^{\delta}(P\|Q) \teq \inf_{\substack{\psi:\cA \to [0,1]:\\  \sum_{a \in \cA}\psi(a) P(a)\geq 1-\delta}} \left| \log \sup_{a \in \cA}\frac{\psi(a)P(a)}  {Q(a)} \right|^{+},
\end{align*}
where $|x|^{+} = \max\{0, x\}$.
\begin{cor}
    \label{cor: bound by smooth max renyi divergence}
    For any set $\cF\subseteq \cA\times \cB\times \cC$, any real numbers $\delta_1,\delta_2 \geq 0$, and any integers $M_1, M_2 \geq 1$, we have
    \begin{align*}
        &\mathrm{E}\left[\mathrm{E}\left[ \mathrm{E}\left[\mathbf{1}\{ (A,\tilde B, \tilde C) \notin \cF\} \big|A, \tilde B \right]^{M_2} \Big| A \right]^{M_1}\right]\\
        &\leq \delta_1 + \delta_2 + \Pr\left\{(A,B,C) \notin \cF\right\}\notag\\
        &\quad +e^{-\exp(- D_{\infty}^{\delta_2}(P_{ABC}\|P_{AB}\times P_{\tilde C|\tilde B}))M_2
          + D_{\infty}^{\delta_1}(P_{AB}\|P_{A}\times P_{\tilde B})}\\
        &\quad + e^{-\exp(-D_{\infty}^{\delta_1}(P_{AB}\|P_{A}\times P_{\tilde B})) M_1}.
    \end{align*}
\end{cor}
\begin{proof}
    For an arbitrarily fixed $\epsilon > 0$, let $\psi_1$ and $\psi_2$ be functions such that
    $\mathrm{E}[\psi_1(A,B)] \geq 1-\delta_1$, $\mathrm{E}[\psi_2(A,B,C)] \geq 1-\delta_2$,
    \begin{align}
        \bar D_{\infty}^{\delta_1}(P_{AB}\|P_{A}\times P_{\tilde B}| \psi_{1})\leq D_{\infty}^{\delta_1}(P_{AB}\|P_{A}\times P_{\tilde B}) + \epsilon,
        \label{equ: cor: proof: bar D1 leq D1}
    \end{align}
    and
    \begin{align}
        &\bar D_{\infty}^{\delta_2}(P_{ABC}\|P_{AB}\times P_{\tilde C|\tilde B}| \psi_{2})\notag\\
        &\leq D_{\infty}^{\delta_2}(P_{ABC}\|P_{AB}\times P_{\tilde C|\tilde B}) + \epsilon,
        \label{equ: cor: proof: bar D2 leq D2}
    \end{align}
    where
    \begin{align*}
        \bar D_{\infty}^{\delta}(P\| Q| \psi) = \left| \log \sup_{a \in \cA}\frac{\psi(a)P(a)}  {Q(a)} \right|^{+}.
    \end{align*}
    Then, we have
    \begin{align}
        &\mathrm{E}[\psi_1(A,B)\psi_2(A,B,C)] \notag\\
        &\geqo{(a)} \mathrm{E}[\psi_1(A,B)]+\mathrm{E}[\psi_2(A,B,C)] - 1 \notag\\
        &\geq 1 - \delta_1 - \delta_2,
        \label{equ: ineq for psi1 and psi2 for renyi divergence}
    \end{align}
    where (a) follows since $x y \geq x + y -1$ for $x,y \in [0, 1]$.

    On the other hand, let $\alpha_1$ and $\alpha_2$ be constants such that
    \begin{align*}
        \alpha_1 &= \exp\left(-\bar D_{\infty}^{\delta_1}(P_{AB}\|P_{A}\times P_{\tilde B}| \psi_{1})\right),\\
        \alpha_2 &= \exp\left(- \bar D_{\infty}^{\delta_2}(P_{ABC}\|P_{AB}\times P_{\tilde C|\tilde B}| \psi_{2})\right).
    \end{align*}
    Then, for any $(a,b,c)\in \cA\times\cB\times\cC$, we have
    \begin{align*}
        &\alpha_1 \psi_1(a,b)P_{AB}(a, b)\\
        & \leq \left(\inf_{(a, b) \in \cA \times \cB} \frac{P_{A}(a) P_{\tilde B}(b)}{\psi_{1}(a, b)P_{AB}(a, b)}\right) \psi_1(a,b)P_{AB}(a, b)\\
        &\leq P_{A}(a) P_{\tilde B}(b),
    \end{align*}
    and
    \begin{align*}
        &\alpha_2 \psi_2(a,b,c)P_{ABC}(a, b, c)\\
        &\leq \left(\inf_{(a, b, c) \in \cA \times \cB \times \cC} \frac{P_{AB}(a, b) P_{\tilde C|\tilde B}(c|b)}{\psi_{2}(a, b, c)P_{ABC}(a, b, c)}\right)\\
        &\quad \times \psi_2(a,b,c)P_{ABC}(a, b, c)\\
        &\leq P_{AB}(a, b) P_{\tilde C|\tilde B}(c|b).
    \end{align*}
    Thus, $\psi_{1}$, $\psi_{2}$, $\alpha_{1}$, and $\alpha_{2}$ satisfy \eqref{equ: alpha_1psi_1condition} and \eqref{equ: alpha_2psi_2condition}.

    Plugging these functions and constants into \eqref{equ: lem: bound for covering lemma}, we have 
    \begin{align*}
        &\mathrm{E}\left[\mathrm{E}\left[ \mathrm{E}\left[\mathbf{1}\{ (A,\tilde B, \tilde C) \notin \cF\} \big|A, \tilde B \right]^{M_2} \Big| A \right]^{M_1}\right]\\
        &\leq \delta_1 + \delta_2 + \Pr\left\{(A,B,C) \notin \cF\right\}\notag\\
        &\quad +e^{-\exp(- D_{\infty}^{\delta_2}(P_{ABC}\|P_{AB}\times P_{\tilde C|\tilde B}) - \epsilon)M_2}\\
        &\quad\times e^{D_{\infty}^{\delta_1}(P_{AB}\|P_{A}\times P_{\tilde B}) + \epsilon}\\
        &\quad +e^{-\exp(- D_{\infty}^{\delta_1}(P_{AB}\|P_{A}\times P_{\tilde B}) - \epsilon) M_1},
    \end{align*}
    where we use inequalities \eqref{equ: cor: proof: bar D1 leq D1}, \eqref{equ: cor: proof: bar D2 leq D2}, and \eqref{equ: ineq for psi1 and psi2 for renyi divergence}. Since $\epsilon > 0$ is arbitrary, this completes the proof.
\end{proof}
\begin{rem}
    The original definition of the smooth max R\'enyi divergence (cf.\ \cite{uyematsu2014ITWrevisiting}) is as follows:
    \mathindent=0mm
    \begin{align*}
        D_{\infty}^{\delta}(P\|Q)^{-} \teq \inf_{\substack{\psi:\cA \to [0,1]:\\  \sum_{a \in \cA}\psi(a) P(a)\geq 1-\delta}} \log \sup_{\cB \subseteq \cA}\frac{\sum_{b \in \cB}\psi(b)P(b)}{\sum_{b \in \cB}Q(b)}.
    \end{align*}
    \mathindent=7mm
    Since for non-negative real valued functions $f(b)$ and $g(b)$, it holds that (cf.\ e.g.\ \cite[Lemma 16.7.1]{cover2006eit})
    \begin{align*}
        \frac{\sum_{b\in\cB} f(b)}{\sum_{b\in\cB} g(b)} \leq \sup_{b\in\cB}\frac{f(b)}{g(b)},
    \end{align*}
    $D_{\infty}^{\delta}(P\|Q)^{-}$ can be simply defined as
    \begin{align*}
        D_{\infty}^{\delta}(P\|Q)^{-} \teq \inf_{\substack{\psi:\cA \to [0,1]:\\  \sum_{a \in \cA}\psi(a) P(a)\geq 1-\delta}} \log \sup_{a \in \cA}\frac{\psi(a)P(a)}  {Q(a)}.
    \end{align*}
    In this definition, it may be a negative value depending on $\delta$. Since this case is meaningless in this study, we adopt $D_{\infty}^{\delta}(P\|Q)$. Here we also note that
    \begin{align}
        D_{\infty}^{\delta}(P\|Q) = \left| D_{\infty}^{\delta}(P\|Q)^{-} \right|^{+}.
        \label{rem: equ: D = |D-|+}
    \end{align}
\end{rem}

\section{Inner and Outer Bounds on the Set of Pairs of Numbers of Codewords}
\label{sec: inner and outer bounds}
In this section, we give outer and inner bounds on $\cM(D,\epsilon| X)$ by using the smooth max R\'enyi divergence.

First of all, we show a bound on the probability of the two events $\cE_{1}(D_1)$ and $\cE_{2}(D_2)$ for the successive refinement problem. In what follows, let $\cU$ be an arbitrary finite or countably infinite set.
\begin{thm}
    \label{thm: existance of code for SR}
    For a source $X$, let $(\tilde U, \tilde Y, \tilde Z) \in \cU\times\cY\times\cZ$ be RVs such that $(\tilde U, \tilde Y, \tilde Z)$ is independent of $X$. Then, for any real numbers $D_1$, $D_2 \geq 0$, there exists a code $(f_1,f_2,\varphi_1,\varphi_2)$ such that numbers of codewords of encoder 1 and encoder 2 are $M_1$ and $M_2$, respectively, and
    \begin{align*}
        &\Pr\{\cE_{1}(D_1)\cup\cE_{2}(D_2)\}\\
        &\leq \mathrm{E}\left[\mathrm{E} \left[ \mathrm{E} \left[\mathbf{1}\{ (X, \tilde U, \tilde Y, \tilde Z) \notin \cD\} \middle| X, \tilde U, \tilde Y \right]^{{M_{2}}} \middle| X \right]^{M_{1}}\right],
    \end{align*}
    where
    \begin{align*}
        \cD &= \left\{(x,u,y,z)\in\cX\times\cU\times\cY\times\cZ:\right.\\
        &\quad \left. d_1(x, y)\leq D_1, d_2(x, z)\leq D_2\right\}.
    \end{align*}    
\end{thm}
\begin{proof}
    We generate $(\tilde u_1, \tilde y_1), (\tilde u_2, \tilde y_2),\cdots,(\tilde u_{M_{1}}, \tilde y_{M_{1}}) \in\cU\times\cY$
    independently subject to the probability distribution $P_{\tilde U\tilde Y}$,
    and define the set $\cC_{1} \teq \{(\tilde u_1, \tilde y_1), (\tilde u_2, \tilde y_2),\cdots,(\tilde u_{M_{1}}, \tilde y_{M_{1}})\}$.
    For any $i\in[1:M_{1}]$, we generate $\tilde z_{i,1},\tilde z_{i,2},\cdots,\tilde z_{i,M_2}\in\cZ$
    independently subject to the probability distribution $P_{\tilde Z|\tilde U\tilde Y}(\cdot|u_{i},y_{i})$,
    and define the set $\cC_{2,i} \teq \{\tilde z_{i,1},\tilde z_{i,2},\cdots,\tilde z_{i,M_2}\}$.
    We denote $\{\cC_{2,1},\cC_{2,2},\cdots,\cC_{2,M_1}\}$ as $\cC_2$.
    For a given set $\cC_{1}$, $\cC_{2}$ and a given symbol $x\in\cX$, we
    choose $i \in [1: M_1]$ and $j \in [1: M_2]$ such that
    \begin{align*}
        d_1(x,\tilde y_{i}) \leq D_1 \mbox{ and }
        d_2(x,\tilde z_{i,j}) \leq D_2.
    \end{align*}
    If there does not exist such pair, we set $(i,j)=(1,1)$.
    For this pair, we define encoders $f_1$ and $f_2$ as
    \begin{align*}
        f_1(x) = i \mbox{ and }
        f_2(x) = j.
    \end{align*}
    On the other hand, we define decoders $\varphi_1$ and $\varphi_2$ as
    \begin{align*}
        \varphi_{1}(i) = \tilde y_{i} \mbox{ and }
        \varphi_{2}(i,j) = \tilde z_{i,j}.
    \end{align*}
    
    By taking the average over the random selection of $C_{1}$ and $C_{2}$,
    the average probability of $\Pr\{\cE_1(D_1)\cup\cE_2(D_2)\}$ is as follows:
    \begin{align*}
        & \mathrm{E}[\Pr\{\cE_1(D_1)\cup\cE_2(D_2)\}]\\
        %
        &= \Pr\left\{\bigcap_{i=1}^{M_{1}} \bigcap_{j=1}^{M_{2}}
            \{d_1(X,\tilde Y_{i})>D_1 \mbox{ or } d_2(X, \tilde Z_{i,j})>D_2\}\right\}\\
        &= \Pr\left\{\bigcap_{i=1}^{M_{1}} \bigcap_{j=1}^{M_{2}}
            \{(X,\tilde U_i, \tilde Y_i, \tilde Z_{i,j}) \notin \cD \}\right\},
    \end{align*}
    where $\{(\tilde U_i, \tilde Y_i, \tilde Z_{i,j})\}$ denote randomly selected sequences in $C_1$ and $C_2$.
    Now, by noting that $\tilde Z_{i,j}$ is generated for a given $(\tilde U_i,\tilde Y_i)$,
    the theorem follows from Lemma \ref{lem: Covering Lemma} by setting that $A=X$, $\tilde B=(\tilde U, \tilde Y)$, and $\tilde C=\tilde Z$.
\end{proof}
\begin{rem}
    This proof is valid even without the RV $\tilde U$. This auxiliary RV is introduced merely for consistency with the outer bound. However, the following intuitive interpretation is possible: $\tilde U$ is partial information of $\tilde Z$ transmitted to two decoders. In an extreme case, if we set $\tilde U = \tilde Z$, $M_2$ has no effect on the bound of Theorem \ref{thm: existance of code for SR}. Hence, we can make $M_2 = 1$. On the other hand, $M_1$ must be increased to satisfy a given probability level. If we set $\tilde U \neq \tilde Z$, $M_1$ may be decreased by increasing $M_2$. In other words, $\tilde U$ manages the balance of the numbers of codewords. This intuition may be useful to set the numbers of codewords in the actual code construction.
\end{rem}

We use the next notation for the sake of simplicity.
\begin{defn}
    For RVs $(A, B, C)$, we define
    \begin{align*}
        I_{\infty}^{\delta}(A; B) &\teq D_{\infty}^{\delta}(P_{AB}\|P_{A}\times P_{B}),\\
        I_{\infty}^{\delta}(A; B| C) &\teq D_{\infty}^{\delta}(P_{ABC}\|P_{AC}\times P_{B|C}).
    \end{align*}
\end{defn}
We also define the following set of probability distributions for a given source $X$ and constants $D$ and $\epsilon$.
\begin{align*}
    \cP(D, \epsilon| X) &\teq \{P_{UYZ|X} \in \cP_{\cU\cY\cZ|\cX}:\\ &\quad \Pr\{d_1(X,Y)>D_1\} \leq \epsilon_1,\\
    &\quad \Pr\{d_2(X,Z)>D_2\} \leq \epsilon_2\}.
\end{align*}
We note that $\cP(D, \epsilon| X)$ depends on the set $\cU$.

Now, by using the above theorem, we give an inner bound on $\cM(D,\epsilon| X)$.
\begin{thm}[Inner bound]
    \label{thm: Cond for M1 M2}
    For a source $X$, real numbers $D_1$, $D_2 \geq 0$ and $\epsilon_1,\epsilon_2 > 0$, and any set $\cU$, we have
    \mathindent=0mm
    \begin{align*}
        \cM(D,\epsilon| X) \supseteq \bigcup_{(\delta, \beta, \gamma) \in \cS(\epsilon)}
        \bigcup_{P_{UYZ|X} \in \cP(D, \gamma| X)} \cM_{\rm I}(\delta, \beta, P_{UYZ|X}),
    \end{align*}
    \mathindent=7mm
    where $\delta = (\delta_{1}, \delta_{2})$, $\beta = (\beta_{1}, \beta_{2})$, $\gamma = (\gamma_{1}, \gamma_{2})$,
    \begin{align}
        &\cS(\epsilon) \teq \left\{(\delta, \beta, \gamma) \in (0, 1]^{6}: \delta_1 + \delta_2 +\gamma_1 + \gamma_2 \right.\notag\\
        &\quad \left. + \beta_1 + \beta_2 \leq \min\{\epsilon_1,\epsilon_2\}\right\},
        \label{equ: cond of d g b e}\\
        &\cM_{\rm I}(\delta, \beta, P_{UYZ|X}) \teq \left\{(M_1, M_2) \in \mathbb{N}^{2}: \right. \notag\\
        &\quad \log M_1 \geq I_{\infty}^{\delta_{1}}(X; U, Y) + \log \log\frac{1}{\beta_1},\notag\\
        &\quad \log M_2 \geq I_{\infty}^{\delta_{2}}(X; Z| U, Y) \notag\notag\\
        &\quad + \left. \log \left(I_{\infty}^{\delta_{1}}(X; U, Y) + \log\frac{1}{\beta_2} \right) \right\},
    \end{align}
    and $(X, U, Y, Z)$ is a tuple of RVs with the probability distribution $P_{X} \times P_{UYZ|X}$.
\end{thm}


\begin{proof}
    We only have to show that $(M,D)$ is $\epsilon$-achievable for $(\delta, \beta, \gamma) \in \cS(\epsilon)$, $P_{UYZ|X} \in \cP(D, \gamma| X)$, and $M_1, M_2 \geq 1$ such that
    \begin{align}
        M_1 &= \left \lceil \exp\left(I_{\infty}^{\delta_{1}}(X; U, Y)\right)\log\frac{1}{\beta_1} \right \rceil,
        \label{equ: cond of M1 M2 for SR 1}\\
        M_2 &= \left \lceil \exp\left(I_{\infty}^{\delta_{2}}(X; Z| U, Y)\right) \left(I_{\infty}^{\delta_{1}}(X; U, Y) + \log\frac{1}{\beta_2} \right) \right \rceil.
        \label{equ: cond of M1 M2 for SR 2}
    \end{align}

    To this end, let $(\tilde U, \tilde Y, \tilde Z) \in \cU \times \cY \times \cZ$ be RVs that is independent of $X$ and has the same marginal distribution as $(U, Y, Z)$, i.e., $P_{\tilde U \tilde Y \tilde Z} = P_{UYZ}$. Then, according to Theorem \ref{thm: existance of code for SR}, there exists a code $(f_1,f_2,\varphi_1, \varphi_2)$ such that numbers of codewords of encoder 1 and encoder 2 are $M_1$ and $M_2$, respectively, and
    \begin{align}
        &\max\{\Pr\{\cE_{1}(D_1)\}, \Pr\{\cE_{2}(D_2)\}\}\notag\\
        &\leq \mathrm{E}\left[\mathrm{E}\left[ \mathrm{E}\left[\mathbf{1}\{ (X, \tilde U, \tilde Y, \tilde Z) \notin \cD\} \middle|X, \tilde U, \tilde Y \right]^{{M_{2}}} \middle| X \right]^{M_{1}}\right].
        \label{equ: achievability bound 1 for SR}
    \end{align}
    
    On the other hand, according to Corollary \ref{cor: bound by smooth max renyi divergence},
    we have
    \begin{align}
        &\mathrm{E}\left[\mathrm{E}\left[ \mathrm{E}\left[\mathbf{1}\{ (X, \tilde U, \tilde Y, \tilde Z) \notin \cD\} \big|X, \tilde U, \tilde Y \right]^{{M_{2}}} \Big| X \right]^{M_{1}}\right]\notag\\
        &\leq \delta_1 + \delta_2 + \Pr\left\{(X, U, Y, Z) \notin \cD\right\}\notag\\
        &\quad +e^{-\exp(- I_{\infty}^{\delta_{2}}(X; Z| U, Y))M_2 + I_{\infty}^{\delta_{1}}(X; U, Y)}\notag\\
        &\quad + e^{-\exp(- I_{\infty}^{\delta_{1}}(X; U, Y)) M_1}\notag\\
        &\overset{\rm (a)}{\leq}  \delta_1 + \delta_2 +\gamma_1 + \gamma_2 + \beta_1 + \beta_2\notag\\
        &\overset{\rm (b)}{\leq}  \min\{\epsilon_1,\epsilon_2\},
        \label{equ: achievability bound 2 for SR}
    \end{align}
    where (a) follows from \eqref{equ: cond of M1 M2 for SR 1}, \eqref{equ: cond of M1 M2 for SR 2} and the fact that $P_{UYZ|X} \in \cP(D, \gamma|X)$, and (b) comes from \eqref{equ: cond of d g b e}. This implies that $(M, D)$ is $\epsilon$-achievable.
\end{proof}
\begin{rem}
    The proof is also valid if we do not restrict $(\tilde U, \tilde Y, \tilde Z)$ to be the same distribution as $(U, Y, Z)$. However, for the sake of simplicity, we consider the restricted case.
\end{rem}

An outer bound on $\cM(D, \epsilon| X)$ is given in the next theorem.
\begin{thm}[Outer bound]
    \label{thm: outer bound}
    For a source $X$, real numbers $D_1$, $D_2 \geq 0$ and $\epsilon_1,\epsilon_2 > 0$, and any set $\cU$ such that $|\cU| \geq |\cX|$, we have
    \begin{align*}
        \cM(D,\epsilon| X) \subseteq \bigcup_{P_{UYZ|X} \in \cP(D, \epsilon| X)} \bigcap_{\delta \in (0, 1]^{2}}\cM_{\mathrm{O}}(\delta, P_{UYZ|X}),
    \end{align*}
    where
    \begin{align*}
        &\cM_{\rm O}(\delta, P_{UYZ|X}) \teq \left\{(M_1, M_2) \in \mathbb{N}^{2}: \right.\\
        &\quad \log M_1\geq I_{\infty}^{\delta_1}(X; U, Y) + \log \delta_1,\\
        &\quad \left. \log M_2 \geq I_{\infty}^{\delta_2}(X; Z| U, Y) + \log \delta_2 \right\}.
    \end{align*}
\end{thm}
\begin{rem}
    \label{rem:LooseOuterBound}
    The RHS of the outer bound can be further bounded as
    \begin{align*}
        &\bigcup_{P_{UYZ|X} \in \cP(D, \epsilon| X)} \bigcap_{\delta \in (0, 1]^{2}}\cM_{\mathrm{O}}(\delta, P_{UYZ|X})\\
        &\subseteq \bigcup_{(\delta, \beta, \gamma) \in \cS(\epsilon)} \bigcup_{P_{UYZ|X} \in \cP(D, \epsilon| X)} \cM_{\mathrm{O}}(\delta, P_{UYZ|X}),
    \end{align*}
    where we note that $\beta$ and $\gamma$ do not affect the bound. Although this bound is looser than that of Theorem \ref{thm: outer bound}, it may be easier to compare with the inner bound of Theorem \ref{thm: Cond for M1 M2}. In fact, it can be immediately noticed that the difference between this loose bound and the inner bound comes from the difference between 
    $\cP(D, \gamma| X)$ and $\cP(D, \epsilon| X)$, and the difference between $\cM_{\rm I}(\delta, \beta, P_{UYZ|X})$ and $\cM_{\mathrm{O}}(\delta, P_{UYZ|X})$.
\end{rem}

Before proving the theorem, we show some necessary lemmas.
\begin{lem}
    \label{lem: Han's lemma for bounded image}
    Suppose that a pair of RVs $(A,B)$ on $\cA\times\cB$ satisfies
    \begin{align}
        |\{a\in\cA: P_{A|B}(a|b)>0\}|\le M,\notag\\
        \quad \forall b\in\cB \mbox{ s.t. } P_{B}(b)>0,
        \label{equ: bound of P_{A|B} by M}
    \end{align}
    for some $M > 0$. Then, for any $\epsilon \in (0, 1]$, we have
    \begin{align*}
        \Pr\left\{P_{A|B}(A|B)\ge \frac{\epsilon}{M}\right\}> 1-\epsilon. 
    \end{align*}
\end{lem}
\begin{proof}
    Since the lemma can be proved in a similar manner as \cite[Lemma 2.6.2]{hanspringerinformation}, we omit the proof.
\end{proof}

The next lemma is an extended version of \cite[Lemma 4]{uyematsu2014ITWrevisiting}, which gives a bound on the size of the image of a function.
\begin{lem}
    \label{lem: lower bound of image of g}
    For a function $g:\cA \ra \cB \times \cC$ and $c \in \cC$, let $\|g\|_{c}$ denote the size of the image of $g$ when one output is fixed to $c$, i.e.,
    \begin{align*}
        \|g\|_{c} = |\{b\in\cB: g(a)=(b,c), \exists a\in\cA\}|.
    \end{align*}
    Then, for any $\delta\in(0,1]$, any RV $A \in \cA$, and $(B, C) = g(A)$, we have
    \begin{align*}
        \log \sup_{c\in\cC}\|g\|_{c} \geq I_{\infty}^{\delta}(A; B| C) + \log \delta.
    \end{align*}
\end{lem}
\begin{proof}
    Let $M = \sup_{c\in\cC}\|g\|_{c}$. Define a subset $\cD_{\delta}\subseteq \cB \times \cC$ and the function $\psi_{o}:\cA\times\cB\times\cC\ra [0,1]$ as
    \mathindent=0mm
    \begin{align}
        \cD_{\delta} &\teq \left\{(b,c)\in\cB\times\cC: P_{B | C}(b | c) \geq \frac{\delta}{M}\right\},
        \label{equ: def of D_delta}\\
        \psi_{o}(a,b,c) &\teq \mathbf{1}\{(a,b,c)\in \cA\times\cD_{\delta}\}.
    \end{align}
    \mathindent=7mm
    Since $P_{BC}(b,c) = \sum_{a\in\cA} P_{A}(a) \mathbf{1}\{(b,c)=g(a)\}$,
    $P_{B | C}(b | c) > 0$ for $c \in \cC$ such that $P_{C}(c) > 0$ if and only if there exists $a\in\cA$ such that $(b,c)=g(a)$ and $P_{A}(a) > 0$. Thus, for $c \in \cC$ such that $P_{C}(c) > 0$, we have    
    \begin{align*}
        &|\{b\in\cB: P_{B | C}(b | c)>0\}\\
        &\quad =  |\{b \in \cB: (b, c)=g(a), \exists a \in\cA \mbox{ s.t. }P_{A}(a)>0\}\\
        &\quad \leq \|g\|_{c}\\
        &\quad \leq M.
    \end{align*}    
    Then, by using Lemma \ref{lem: Han's lemma for bounded image}, it is easy to see that
    \begin{align*}
        &\sum_{(a,b,c)\in\cA\times\cB\times\cC} \psi_{o}(a,b,c) P_{ABC}(a,b,c)\\
        &=\Pr\left\{P_{B| C}(B| C)\geq \frac{\delta}{M}\right\}\\
        & > 1 - \delta.
    \end{align*}

    Thus, we have
    \begin{align*}
        &I_{\infty}^{\delta}(A; B| C)\\
        &\leq \left| \log \sup_{(a,b,c)\in \cA\times\cB\times\cC}\frac{\psi_{o}(a,b,c)P_{ABC}(a,b,c)}
            {P_{AC}(a, c)P_{B| C}(b| c)} \right|^{+}\\
        &= \left| \log \sup_{(a,b,c)\in \cA\times\cD_{\delta}}\frac{P_{B|AC}(b|a, c)}{P_{B| C}(b| c)} \right|^{+}\\
        &\leqo{(a)} \left| \log \sup_{(a,b,c)\in \cA\times\cD_{\delta}}\frac{P_{B|AC}(b|a, c)}{\delta/M} \right|^{+}\\
        &\leq \log M - \log \delta,
    \end{align*}
    where (a) comes from the definition \eqref{equ: def of D_delta}. This completes the proof.
\end{proof}
\begin{rem}
    \label{rem: upper bound of smooth renyi mutual information}
    For a triple of RVs $(A, B, C)$ on $\cA \times \cB \times \cC$, let $M = |\cB|$. Then, in the same way as the above proof, we can easily show that
    \begin{align*}
        I_{\infty}^{\delta}(A; B| C) \leq \log |\cB| - \log \delta.
    \end{align*}
\end{rem}

Now, we give the proof of Theorem \ref{thm: outer bound}.
\begin{proof}[{Proof of Theorem \ref{thm: outer bound}}]
    Let $\|f_{1}\|$ be the size of the image of an encoder $f_{1}$. Since $\|f_{1}\| \leq |\cX|$ and $|\cX| \leq |\cU|$ by the assumption, there exists an injective function $\mathrm{id}: [1: \|f_{1}\|] \ra \cU$. For this function, let $\cU_{\mathrm{id}} \subseteq \cU$ be the image of $\mathrm{id}$ and $\mathrm{id}^{-1}: \cU_{\mathrm{id}} \ra [1: \|f_{1}\|]$ be the inverse function of $\mathrm{id}$ on $\cU_{\mathrm{id}}$.
    
    Suppose that $(M,D)$ is $\epsilon$-achievable. Then, there exists a code $(f_1,f_2,\varphi_1,\varphi_2)$ such that
    \begin{align*}
        \Pr\left\{d_1(X,\varphi_1(f_1(X)))>D_1\right\} \leq \epsilon_1,\\
        \Pr\left\{d_2(X,\varphi_2(f_1(X),f_2(X)))>D_2\right\} \leq \epsilon_2.
    \end{align*}
    Thus, by setting $U=\mathrm{id}(f_1(X))$, $Y=\varphi_1(f_1(X))$, and $Z=\varphi_2(f_1(X),f_2(X))$, we have
    \begin{align}
        \Pr\left\{d_1(X,Y)>D_1\right\} \leq \epsilon_1,
        \label{equ: outer bound cond of Y and D1}\\
        \Pr\left\{d_2(X,Z)>D_2\right\} \leq \epsilon_2.
        \label{equ: outer bound cond of Z and D2}
    \end{align}
    
    For a constant value $c$, let $g_1(x) = (\mathrm{id}(f_1(x)), \varphi_1(f_1(x)), c)$, $A = X$, $B = (U, Y)$, and $C = c$. According to Lemma \ref{lem: lower bound of image of g}, for any $\delta_{1} \in (0, 1]$, we have
    \begin{align}
        \log \|g_1\|_{c}
        &\geq I_{\infty}^{\delta_1}(X; U, Y| C) + \log \delta_1\notag\\
        &= I_{\infty}^{\delta_1}(X; U, Y) + \log \delta_1.
        \label{equ: log g1 geq D1}
    \end{align}
    On the other hand, we have
    \begin{align}
        \|g_1\|_{c}
        &= |\{(u,y)\in\cU_{\mathrm{id}}\times\cY: g_1(x)=(u, y, c), \exists x\in\cX\}|\notag\\
        &= \sum_{u \in \cU_{\mathrm{id}}} \sum_{y \in \cY} \mathbf{1}\{g_1(x)=(u, y, c), \exists x\in\cX\}\notag\\
        &= \sum_{u \in \cU_{\mathrm{id}}} \mathbf{1}\{g_1(x)=(u, \varphi_{1}(\mathrm{id}^{-1}(u)), c), \exists x\in\cX\}\notag\\
        &\leq M_1,
        \label{equ: g1 leq M1}
    \end{align}
    where the last inequality follows since the size of $\cU_{\mathrm{id}}$ is at most $\|f_{1}\|$ and $\|f_{1}\| \leq M_{1}$. Combining \eqref{equ: log g1 geq D1} and \eqref{equ: g1 leq M1}, we have
    \begin{align}
        \log M_1 \geq I_{\infty}^{\delta_1}(X; U, Y) + \log \delta_1.
        \label{euq: log M1 geq D1}
    \end{align}

    Let $g_2(x) = (\varphi_2(f_1(x),f_2(x)), \mathrm{id}(f_1(x)), \varphi_1(f_1(x)))$, $A = X$, $B = Z$, and $C = (U,Y)$. Then, according to Lemma \ref{lem: lower bound of image of g}, for any $\delta_{2} \in (0, 1]$, we have
    \begin{align}
        \log \sup_{(u,y)\in\cU\times\cY}\|g_2\|_{(u,y)} \geq I_{\infty}^{\delta_2}(X; Z| U, Y) + \log \delta_2.
        \label{equ: log g2 geq D2}
    \end{align}
    On the other hand, for any $(u,y)\in\cU_{\mathrm{id}}\times\cY$, we have
    \begin{align}
        \|g_2\|_{(u,y)}
        &= \sum_{z\in\cZ}
        \mathbf{1}\{z=\varphi_2(\mathrm{id}^{-1}(u),f_2(x)), \notag\\
        &\quad \mathrm{id}^{-1}(u) = f_1(x), y = \varphi_1(\mathrm{id}^{-1}(u)), \exists x\in\cX\}\notag\\
        &\leq \sum_{z\in\cZ}
        \mathbf{1}\{\exists x\in\cX, z=\varphi_2(\mathrm{id}^{-1}(u),f_2(x))\}\notag\\
        &\leq \sum_{z\in\cZ}
        \mathbf{1}\{\exists j\in[1:M_2], z=\varphi_2(\mathrm{id}^{-1}(u),j)\}\notag\\
        &\leq \sum_{j\in[1:M_2]}\sum_{z\in\cZ}\mathbf{1}\{z=\varphi_2(\mathrm{id}^{-1}(u),j)\}\notag\\
        &= M_2.
        \label{equ: g2 leq M2}
    \end{align}
    We note that for any $(u,y)\in \{\cU\setminus \cU_{\mathrm{id}}\} \times \cY$, it holds that $\|g_2\|_{(u,y)} = 0$. Combining \eqref{equ: log g2 geq D2} and \eqref{equ: g2 leq M2}, we have
    \begin{align}
        \log M_2 \geq I_{\infty}^{\delta_2}(X; Z| U, Y) + \log \delta_2.
        \label{euq: log M2 geq D2}
    \end{align}
    Since $\delta_{1} \in (0, 1]$ and $\delta_{2} \in (0, 1]$ are arbitrary, \eqref{euq: log M1 geq D1} and \eqref{euq: log M2 geq D2} imply that
    \begin{align*}
        (M_1, M_2) \in \bigcap_{\delta \in (0, 1]^{2}}\cM_{\mathrm{O}}(\delta, P_{UYZ|X}).
    \end{align*}

    Now, by recalling that $(X, U, Y, Z)$ satisfy \eqref{equ: outer bound cond of Y and D1} and \eqref{equ: outer bound cond of Z and D2}, for any $\epsilon$-achievable pair $(M, D)$, we have
    \begin{align*}
        (M_1, M_2) \in \bigcup_{P_{UYZ|X} \in \cP(D, \epsilon| X)} \bigcap_{\delta \in (0, 1]^{2}}\cM_{\mathrm{O}}(\delta, P_{UYZ|X}).
    \end{align*}
    This completes the proof.
\end{proof}

\begin{rem}
    If we do not employ the RV $U$ which has a role in fixing the RV $f_{1}(X)$ to a certain codeword, we cannot bound $\|g_2\|_{(u,y)}$ by $M_{2}$ in \eqref{equ: g2 leq M2}. Thus in this proof, introducing $U$ is quite important.
\end{rem}

\begin{rem}
    In \cite{matsuta2015nab}, we gave inner and outer bounds on $\cM(D, \epsilon| X)$ by using the $\alpha$-mutual information of order infinity \cite{verdu2015alpha}, where the $\alpha$-mutual information is a generalized version of the mutual information. In this paper, however, we use the smooth max R\'enyi divergence. This is because it is compatible with the information spectrum quantity which is well studied and useful to analyze rates of a code.
\end{rem}

Finally, we discuss the difference between our inner and outer bounds of Theorems \ref{thm: Cond for M1 M2} and \ref{thm: outer bound}.

If cardinalities of sets $\cX$, $\cY$, $\cZ$, $\cU$ are small, the RHSs of the outer bound and the inner bound may be given by computing their boundaries. Thus, the difference between these two bounds is actually evaluated. On the other hand, if cardinalities of the sets are large, it is difficult to compute their boundaries. However, we can evaluate the difference roughly.

Let $|\cU| = |\cX|$ for the sake of simplicity. Since the main interest is in the case where $\epsilon_1, \epsilon_2 > 0$ are sufficiently small, we assume that for a small real number $\rho \in (0, 1/2]$, $\epsilon_1 = \epsilon_2 = \rho$. Then, we can set that $\delta_i = \beta_i = \gamma_i \approx \epsilon_i$ $(i = 1, 2)$ in the inner bound. Thus, the RHS of the inner bound can be approximated as
\begin{align*}
    &\bigcup_{(\delta, \beta, \gamma) \in \cS(\epsilon)}
    \bigcup_{P_{UYZ|X} \in \cP(D, \gamma| X)} \cM_{\rm I}(\delta, \beta, P_{UYZ|X})\\
    &\approx \bigcup_{P_{UYZ|X} \in \cP(D, \epsilon| X)} \cM_{\rm I}(\epsilon, \epsilon, P_{UYZ|X}).
\end{align*}    
On the other hand, the RHS of the outer bound can be bounded as
\begin{align*}
    &\bigcup_{P_{UYZ|X} \in \cP(D, \epsilon| X)} \bigcap_{\delta \in (0, 1]^{2}}\cM_{\mathrm{O}}(\delta, P_{UYZ|X})\\
    &\subseteq \bigcup_{P_{UYZ|X} \in \cP(D, \epsilon| X)} \cM_{\mathrm{O}}(\epsilon, P_{UYZ|X}).        
\end{align*}    
Thus, the difference between the outer bound and the inner bound can be evaluated by the difference between $\cM_{\rm I}(\epsilon, \epsilon, P_{UYZ|X})$ and $\cM_{\mathrm{O}}(\epsilon, P_{UYZ|X})$ for each fixed $P_{UYZ|X} \in \cP(D, \epsilon| X)$.

According to the definitions of $\cM_{\rm I}(\epsilon, \epsilon, P_{UYZ|X})$ and $\cM_{\mathrm{O}}(\epsilon, P_{UYZ|X})$, the difference comes from boundary pairs $(M_{\textrm{I}, 1}, M_{\textrm{I}, 2})$ of $\cM_{\rm I}(\epsilon, \epsilon, P_{UYZ|X})$ and $(M_{\textrm{O},1}, M_{\textrm{O},2})$ of $\cM_{\mathrm{O}}(\epsilon, P_{UYZ|X})$, where
\mathindent=0mm
\begin{align*}        
    M_{\textrm{I},1} &= \left \lceil \exp \left( I_{\infty}^{\rho}(X; U, Y) + \log \log\frac{1}{\rho}\right) \right \rceil,\\
    M_{\textrm{I},2} &= \left \lceil \exp \left( I_{\infty}^{\rho}(X; Z| U, Y) + \log \left(I_{\infty}^{\rho}(X; U, Y) + \log\frac{1}{\rho} \right) \right) \right \rceil,\\
    M_{\textrm{O},1} &= \left \lceil \exp \left( I_{\infty}^{\rho}(X; U, Y) + \log \rho \right) \right \rceil,\\
    M_{\textrm{O},2} &= \left \lceil \exp \left( I_{\infty}^{\rho}(X; Z| U, Y) + \log \rho \right) \right \rceil.
\end{align*}
\mathindent=7mm
Clearly, the difference between $M_{\textrm{I},1}$ and $M_{\textrm{O},1}$ can be evaluated by second terms of exponents, i.e, $\log \log\frac{1}{\rho} - \log \rho$. Similarly, the difference between $M_{\textrm{I},2}$ and $M_{\textrm{O},2}$ can be evaluated by second terms of exponents, i.e, $\log \left(I_{\infty}^{\rho}(X; U, Y) + \log\frac{1}{\rho} \right) - \log \rho$. Since $I_{\infty}^{\rho}(X; U, Y) \geq 0$, the differences are at most $\log \left(I_{\infty}^{\rho}(X; U, Y) + \log\frac{1}{\rho} \right) - \log \rho$.

Especially, for a finite set $\cX$, since $I_{\infty}^{\rho}(X; U, Y) \leq \log |\cX| - \log \rho$ (see Remark \ref{rem: upper bound of smooth renyi mutual information}), the differences are at most $\log \left(\log|\cX| + 2\log\frac{1}{\rho} \right) - \log \rho$. Furthermore, when we regard $\cX$ as an $n$-fold Cartesian product $\cX^n$, the differences of exponents are at most $\log \left(n \log|\cX| + 2\log\frac{1}{\rho} \right) - \log \rho$. Dividing it by $n$, it obviously vanishes as $n$ tends to infinity. This implies that the RHSs of the inner bound and the outer bound asymptotically coincide with each other in terms of \textit{rate} (which is the exponent of the number of codewords divided by $n$) as $n$ tends to infinity. In fact, in the next section, the rate-distortion region can be given by using our inner and outer bounds.

\section{General Formula for the Rate-Distortion Region}
\label{sec: the limit of the rate region}
In this section, we deal with the coding for an $n$-length source sequence and give a general formula for the rate-distortion region.

First of all, we introduce the spectral (conditional) sup-mutual information rate \cite{hanspringerinformation}.
\begin{defn}
    \label{defn: spectral sup-mutual information rate}
    For a sequence $(\bfX, \bfY, \bfZ) = \{(X^n, Y^n, Z^n)\}_{n = 1}^{\infty}$ of RVs, we define
    \begin{align*}
        \overline{I}(\bfX; \bfY) &\teq \plimsup_{n \ra \infty} \frac{1}{n} \log \frac{P_{Y^n|X^n}(Y^n|X^n)}{P_{Y^n}(Y^n)},\\
        \overline{I}(\bfX; \bfY| \bfZ) &\teq \plimsup_{n \ra \infty} \frac{1}{n} \log \frac{P_{Y^n|X^n Z^n}(Y^n|X^n, Z^n)}{P_{Y^n| Z^n}(Y^n| Z^n)}.
    \end{align*}    
\end{defn}

The smooth max R\'enyi divergence is related to the spectral sup-mutual information rate as shown in the corollary of the next lemma.
\begin{lem}
    \label{lem: smooth max renyi divergense is spectrum divergence}
    Consider two sequences $\bfX$ and $\bfY$ of RVs over the same set. Then, we have
    \mathindent=0mm
    \begin{align*}
        \lim_{\delta\downarrow 0}\limsup_{n\ra\infty} \frac{1}{n} D_{\infty}^{\delta}(P_{X^n} \| P_{Y^n}) = \plimsup_{n\ra\infty} \frac{1}{n} \log \frac{P_{X^n}(X^n)}{P_{Y^n}(X^n)}.
    \end{align*}
    \mathindent=7mm
\end{lem}
\begin{proof}
    For a sequence $\{a_n\}_{n = 1}^{\infty}$ of real numbers, it holds that $\limsup_{n \ra \infty} |a_{n}|^{+} = | \limsup_{n \ra \infty} a_{n}|^{+}$. Thus, according to \eqref{rem: equ: D = |D-|+}, we have
    \begin{align}
        &\lim_{\delta\downarrow 0}\limsup_{n\ra\infty} \frac{1}{n} D_{\infty}^{\delta}(P_{X^n} \| P_{Y^n})\notag\\
        &\quad = \left|\lim_{\delta\downarrow 0} \limsup_{n\ra\infty} \frac{1}{n} D_{\infty}^{\delta}(P_{X^n} \| P_{Y^n})^{-} \right|^{+}.
        \label{equ: lim D = |lim D-|+}
    \end{align}
    According to \cite[Lemma 3]{6691300}, it holds that
    \begin{align}
        &\lim_{\delta\downarrow 0}\limsup_{n\ra\infty} \frac{1}{n} D_{\infty}^{\delta}(P_{X^n} \| P_{Y^n})^{-}\notag\\
        &\quad = \plimsup_{n\ra\infty} \frac{1}{n} \log \frac{P_{X^n}(X^n)}{P_{Y^n}(X^n)}.
        \label{equ: D- = plimsup}
    \end{align}
    Furthermore, according to \cite[Lemma 3.2.1]{hanspringerinformation}, the RHS of \eqref{equ: D- = plimsup} is non-negative. Thus, by combining \eqref{equ: lim D = |lim D-|+} and \eqref{equ: D- = plimsup}, we have the lemma.    
\end{proof}
\begin{cor}
    \label{cor: smooth max renyi divergense is spectrum divergence}
    For a sequence $(\bfX, \bfY, \bfZ)$ of RVs, we have
    \mathindent=0mm
    \begin{align*}
        \lim_{\delta\downarrow 0}\limsup_{n\ra\infty} \frac{1}{n} I_{\infty}^{\delta}(X^n; Y^n) &= \overline{I}(\bfX; \bfY),\\
        \lim_{\delta\downarrow 0}\limsup_{n\ra\infty} \frac{1}{n} I_{\infty}^{\delta}(X^n; Y^n| Z^n) &= \overline{I}(\bfX; \bfY| \bfZ).
    \end{align*}
    \mathindent=7mm
\end{cor}
\begin{proof}
    Since this corollary immediately follows from Lemma \ref{lem: smooth max renyi divergense is spectrum divergence} and Definition \ref{defn: spectral sup-mutual information rate}, we omit the proof.    
\end{proof}

Let $\bfP_{\bfU\bfY\bfZ|\bfX}$ be a sequence of conditional probability distributions $P_{U^nY^nZ^n|X^n} \in \cP_{\cU^n\cY^n\cZ^n|\cX^n}$. We define
\begin{align*}
    \cP_{\mathrm{G}}(D| \bfX) &\teq \{\bfP_{\bfU\bfY\bfZ|\bfX}: \overline{D}_{1}(\bfX,\bfY) \leq D_1,\\ &\quad \overline{D}_{2}(\bfX,\bfZ) \leq D_2\},\\
    \overline{D}_{1}(\bfX,\bfY) &\teq \plimsup_{n \ra \infty} d_1^{(n)}(X^n, Y^n),\\
    \overline{D}_{2}(\bfX,\bfZ) &\teq \plimsup_{n \ra \infty} d_2^{(n)}(X^n, Z^n),\\
    \cR_{\mathrm{G}}(\bfP_{\bfU\bfY\bfZ|\bfX}| \bfX) &\teq \{(R_1,R_2) \in \mathbb{R}^{2}: R_1\geq \overline{I}(\bfX;\bfU, \bfY),\\
    &\quad R_2\geq \overline{I}(\bfX;\bfZ|\bfU, \bfY)\},
\end{align*}
where $(\bfX, \bfU, \bfY, \bfZ)$ is a sequence of RVs $(X^n, U^n, Y^n, Z^n)$ induced by $\bfP_{\bfU\bfY\bfZ|\bfX}$ and a general source $\bfX$. We note that $\cP_{\mathrm{G}}(D| \bfX)$ depends on $\cU$.

The main result of this section is the next theorem which gives a general formula for the rate-distortion region.
\begin{thm}
    \label{thm: general formula for the rate-distortion region}
    For a general source $\bfX$, real numbers $D_{1}, D_{2} \geq 0$, and any set $\cU$ such that $|\cU| \geq |\cX|$, we have
    \begin{align}
        \cR(D| \bfX) &= \bigcup_{\substack{\bfP_{\bfU\bfY\bfZ|\bfX} \in \cP_{\mathrm{G}}(D| \bfX)}} \cR_{\mathrm{G}}(\bfP_{\bfU\bfY\bfZ|\bfX}| \bfX).
        \label{equ: thm: RDR = general formula}
    \end{align}
\end{thm}
\begin{rem}
    We can show that the RHS of \eqref{equ: thm: RDR = general formula} is a closed set by using the diagonal line argument (cf.\ \cite[Remark 5.7.5]{hanspringerinformation}).
\end{rem}
\begin{rem}
    We are not sure whether a sequence $\bfU$ of auxiliary RVs is really necessary or not. It may be possible to characterize the region without it.
\end{rem}

The proof of this theorem is presented in the subsequent two sections. In these sections, for a code, we denote
\begin{align*}
    \hat Y^n &= \varphi_1^{(n)}(f_1^{(n)}(X^n)),\\
    \hat Z^n &= \varphi_2^{(n)}(f_1^{(n)}(X^n),f_2^{(n)}(X^n)).
\end{align*}

\subsection{Direct Part}
In this section, we show
\begin{align}
    \cR(D| \bfX) &\supseteq \bigcup_{\bfP_{\bfU\bfY\bfZ|\bfX} \in \cP_{\mathrm{G}}(D| \bfX)} \cR_{\mathrm{G}}(\bfP_{\bfU\bfY\bfZ|\bfX}| \bfX).
    \label{equ: general direct part}
\end{align}

Let $\bfP_{\bfU\bfY\bfZ|\bfX} \in \cP_{\mathrm{G}}(D| \bfX)$ and suppose that $\overline{I}(\bfX; \bfU, \bfY) < \infty$ and $\overline{I}(\bfX; \bfZ| \bfU, \bfY) < \infty$.
Then, for any $\epsilon, \epsilon_1, \epsilon_2 > 0$ such that $\epsilon_1 = \epsilon_2 = \epsilon$, any $(\delta, \beta, \gamma) \in \cS(\epsilon)$ in \eqref{equ: cond of d g b e}, and every sufficiently large $n$, we have
\begin{align*}
    \Pr\{d_1^{(n)}(X^n,Y^n) > D_1 + \epsilon\} &\leq \gamma_1,\\
    \Pr\{d_2^{(n)}(X^n,Z^n) > D_2 + \epsilon\} &\leq \gamma_2.
\end{align*}
Hence, according to Theorem \ref{thm: Cond for M1 M2}, there exists a sequence of codes $\{(f_1^{(n)}, f_2^{(n)}, \varphi_1^{(n)}, \varphi_2^{(n)})\}$ such that for sufficiently large $n$,
\begin{align*}
    \Pr\{d_1^{(n)}(X^n, \hat Y^n)>D_1 + \epsilon\}&\leq \epsilon,\\
    \Pr\{d_2^{(n)}(X^n, \hat Z^n)>D_2 + \epsilon\} &\leq \epsilon,
\end{align*}
and
\begin{align*}
    \frac{1}{n}\log M_1^{(n)}
    &= \frac{1}{n}\log \left \lceil \exp\left(I_{\infty}^{\delta_{1}}(X^n; U^n, Y^n)\right)\log\frac{1}{\beta_1} \right \rceil,\\
    \frac{1}{n}\log M_2^{(n)}
    &= \frac{1}{n} \log \left \lceil \exp\left(I_{\infty}^{\delta_{2}}(X^n; Z^n| U^n, Y^n)\right)\right. \notag\\
    &\quad \times \left. \left(I_{\infty}^{\delta_{1}}(X^n; U^n, Y^n) + \log\frac{1}{\beta_2} \right) \right \rceil.
\end{align*}

Thus, we have
\begin{align*}
    \limsup_{n\ra\infty} R_{1}^{(n)}
    &\leq \limsup_{n\ra\infty}\frac{1}{n}I_{\infty}^{\delta_1}(X^n; U^n, Y^n)\\
    &\leqo{(a)} \overline{I}(\bfX; \bfU, \bfY),
\end{align*}
and
\begin{align*}
    & \limsup_{n\ra\infty} R_{2}^{(n)}\\
    &\leq \limsup_{n\ra\infty} \frac{1}{n} \log \left(\exp\left(I_{\infty}^{\delta_{2}}(X^n; Z^n| U^n, Y^n)\right)\right. \notag\\
    &\quad \times \left. \left(I_{\infty}^{\delta_{1}}(X^n; U^n, Y^n) + \log\frac{1}{\beta_2} + 1\right) \right)\\
    &\leq \limsup_{n\ra\infty} \frac{1}{n} \log \left(\exp\left(I_{\infty}^{\delta_{2}}(X^n; Z^n| U^n, Y^n)\right)\right. \notag\\
    &\quad \times \left. n\left(\limsup_{n \ra \infty}\frac{1}{n} I_{\infty}^{\delta_{1}}(X^n; U^n, Y^n) + \delta_{2} \right)\right)\\
    &\leqo{(a)} \limsup_{n\ra\infty} \frac{1}{n} \log \left(\exp\left(I_{\infty}^{\delta_{2}}(X^n; Z^n| U^n, Y^n)\right)\right.\notag\\
    &\quad \times \left. n\left(\overline{I}(\bfX; \bfU, \bfY) + \delta_{2} \right)\right)\\
    &\eqo{(b)} \limsup_{n\ra\infty}\frac{1}{n}I_{\infty}^{\delta_2}(X^n; Z^n| U^n, Y^n)\\
    &\leqo{(a)} \overline{I}(\bfX; \bfZ|\bfU, \bfY),
\end{align*}
where (a) comes from Corollary \ref{cor: smooth max renyi divergense is spectrum divergence} and the fact that $I_{\infty}^{\delta}$ is a non-increasing function of $\delta$, and (b) follows since $\overline{I}(\bfX; \bfU, \bfY) < \infty$.

Now, by using the usual diagonal line argument, we can construct a sequence of codes $\{(f_1^{(n)}, f_2^{(n)}, \varphi_1^{(n)}, \varphi_2^{(n)})\}$ such that
\begin{align*}
    \limsup_{n\ra\infty} R_1^{(n)}
    &\leq \overline{I}(\bfX; \bfU, \bfY),\\
    \limsup_{n\ra\infty} R_2^{(n)} 
    &\leq \overline{I}(\bfX; \bfZ|\bfU, \bfY).
\end{align*}
and for any $\epsilon > 0$,
\begin{align*}
    \lim_{n\ra\infty}\Pr\{d_1^{(n)}(X^n,\hat Y^n)>D_1 + \epsilon\}
    &=0,\\
    \lim_{n\ra\infty}\Pr\{d_2^{(n)}(X^n,\hat Z^n)>D_2 + \epsilon\}
    &=0.
\end{align*}
This implies that
\begin{align*}
    \plimsup_{n\ra\infty} d_1^{(n)}(X^n,\hat Y^n)
    &\leq D_1,\\
    \plimsup_{n\ra\infty} d_2^{(n)}(X^n,\hat Z^n)
    &\leq D_2.
\end{align*}
Thus, for any $\bfP_{\bfU\bfY\bfZ|\bfX} \in \cP_{\mathrm{G}}(D| \bfX)$ such that $\overline{I}(\bfX; \bfU, \bfY) \in \mathbb{R}$ and $\overline{I}(\bfX; \bfZ| \bfU, \bfY) \in \mathbb{R}$, we have
\begin{align*}
    (\overline{I}(\bfX; \bfU, \bfY), \overline{I}(\bfX; \bfZ|\bfU, \bfY)) \in \cR(D|\bfX).   
\end{align*}
This implies \eqref{equ: general direct part}.

\subsection{Converse Part}
In this section, we show that
\begin{align}
    \cR(D| \bfX) &\subseteq \bigcup_{\bfP_{\bfU\bfY\bfZ|\bfX} \in \cP_{\mathrm{G}}(D| \bfX)} \cR_{\mathrm{G}}(\bfP_{\bfU\bfY\bfZ|\bfX}| \bfX).
    \label{equ: general converse part}
\end{align}

Suppose that $(R, D)$ is fm-achievable. Then, there exists a sequence of codes satisfying
\begin{align*}
    \plimsup_{n\ra\infty} d_1^{(n)}(X^n,\hat Y^n)
    &\leq D_1,\\
    \plimsup_{n\ra\infty} d_2^{(n)}(X^n,\hat Z^n)
    &\leq D_2,
\end{align*}
and
\begin{align}
    \limsup_{n \ra \infty} \frac{1}{n} \log M_i^{(n)} \leq R_i, \quad \forall i \in \{1, 2\}.
    \label{equ: converse log M leq R}
\end{align}
Thus, we have for any $\gamma>0$
\begin{align*}
    \lim_{n\ra\infty}\Pr\{d_1^{(n)}(X^n,\hat Y^n)>D_1+\gamma\} &=0,\\
    \lim_{n\ra\infty}\Pr\{d_2^{(n)}(X^n,\hat Z^n)>D_2+\gamma\} &=0.
\end{align*}
This implies that there exists a sequence $\{\gamma_{n}\}_{n = 1}^{\infty}$ such that $\lim_{n \ra \infty} \gamma_{n} = 0$ and
\begin{align*}
    \Pr\{d_1^{(n)}(X^n,\hat Y^n)>D_1+\gamma_{n}\} \leq \gamma_{n},\\
    \Pr\{d_2^{(n)}(X^n,\hat Z^n)>D_2+\gamma_{n}\} \leq \gamma_{n}.
\end{align*}
This means that $(M^{(n)}, D + \gamma_{n})$ is $\gamma_{n}$-achievable. According to Theorem \ref{thm: outer bound}, there exists a sequence $\bfP_{\bfU\bfY\bfZ|\bfX} = \{P_{U^nY^nZ^n|X^n}\}_{n = 1}^{\infty}$ of conditional probability distributions such that $P_{U^nY^nZ^n|X^n} \in \cP(D + \gamma_{n}, \gamma_{n}| X^n)$ and for any $\delta \in (0, 1]^2$,
\begin{align*}
    \limsup_{n \ra \infty} \frac{1}{n}\log M_1^{(n)}
    &\geq \limsup_{n \ra \infty} \frac{1}{n}I_{\infty}^{\delta}(X^n; U^n, Y^n),\\
    \limsup_{n \ra \infty} \frac{1}{n}\log M_2^{(n)}
    &\geq \limsup_{n \ra \infty} \frac{1}{n}I_{\infty}^{\delta}(X^n; Z^n| U^n, Y^n).
\end{align*}
Since this holds for any $\delta \in (0, 1]^2$, we have
\begin{align}
    \limsup_{n\ra\infty}\frac{1}{n}\log M_1^{(n)}
    &\geq \lim_{\delta \downarrow 0} \limsup_{n\ra\infty}\frac{1}{n}I_{\infty}^{\delta}(X^n; U^n, Y^n)\notag\\
    &\eqo{(a)} \overline{I}(\bfX; \bfU, \bfY),
    \label{equ: converse M1 geq supectrum I}\\
    \limsup_{n\ra\infty}\frac{1}{n}\log M_2^{(n)}
    &\geq \lim_{\delta \downarrow 0} \limsup_{n\ra\infty}\frac{1}{n}I_{\infty}^{\delta}(X^n; Z^n| U^n, Y^n)\notag\\
    &\eqo{(a)} \overline{I}(\bfX; \bfZ| \bfU, \bfY).
    \label{equ: converse M2 geq supectrum I}
\end{align}
where (a) are comes from Corollary \ref{cor: smooth max renyi divergense is spectrum divergence}. On the other hand, since $P_{U^nY^nZ^n|X^n} \in \cP(D + \gamma_{n}, \gamma_{n}| X^n)$, $\bfP_{\bfU\bfY\bfZ|\bfX}$ must satisfy $\overline{D}_{1}(\bfX; \bfY) \leq D_{1}$ and $\overline{D}_{2}(\bfX; \bfZ) \leq D_{2}$, i.e., $\bfP_{\bfU\bfY\bfZ|\bfX} \in \cP_{\mathrm{G}}(D| \bfX)$.

By combining \eqref{equ: converse log M leq R}, \eqref{equ: converse M1 geq supectrum I}, and \eqref{equ: converse M2 geq supectrum I}, we can conclude that for any fm-achievable pair $(R, D)$,
\begin{align*}
    (R_{1}, R_{2}) &\in \bigcup_{\bfP_{\bfU\bfY\bfZ|\bfX} \in \cP_{\mathrm{G}}(D| \bfX)} \cR_{\mathrm{G}}(\bfP_{\bfU\bfY\bfZ|\bfX}| \bfX).
\end{align*}
This implies \eqref{equ: general converse part}.

\subsection{Discrete Stationary Memoryless Sources}
In this section, we show that the rate-distortion region given in Theorem \ref{thm: general formula for the rate-distortion region} coincides with the region by Rimoldi \cite{272493} when a source $\bfX$ is a discrete stationary memoryless source.

Let $\cX$, $\cY$, and $\cZ$ be finite sets. Since $\bfX = \{X^n\}_{n = 1}^{\infty}$ is a discrete stationary memoryless source, we assume that $X^n = (X_{1}, X_{2}, \cdots, X_{n})$ is a sequence of independent copies of an RV $X$ on $\cX$. We also assume that distortion measures $d_{1}^{(n)}$ and $d_{2}^{(n)}$ are additive, i.e., for two functions $d_{1}: \cX \times \cY \ra [0, +\infty)$ and $d_{2}: \cX \times \cZ \ra [0, +\infty)$, distortion measures are represented as
\begin{align*}
    d_{1}^{(n)}(x^n, y^n) &= \frac{1}{n} \sum_{i = 1}^{n} d_{1}(x_{i}, y_{i}),\\
    d_{2}^{(n)}(x^n, z^n) &= \frac{1}{n} \sum_{i = 1}^{n} d_{2}(x_{i}, z_{i}).
\end{align*}
We define
\begin{align*}
    \cP_{\mathrm{M}}(D| X) &\teq \{P_{YZ|X} \in \cP_{\cY \cZ| \cX}: \mathrm{E}[d_{1}(X,Y)] \leq D_1,\\
    &\quad \mathrm{E}[d_{2}(X,Z)] \leq D_2\},
\end{align*}
and for $P_{YZ|X} \in \cP_{\cY \cZ| \cX}$,
\begin{align*}
    \cR_{\mathrm{M}}(P_{YZ|X}| X) &\teq \{(R_1,R_2) \in \mathbb{R}_{\geq 0}^{2}: R_1\geq I(X;Y),\notag\\
    &\quad R_1+R_2 \geq I(X;Y,Z)\},
\end{align*}
where $(X, Y, Z)$ is the tuple of RVs induced by a conditional probability distribution $P_{YZ|X} \in \cP_{\cY\cZ|\cX}$ and a given RV $X$. Then, we have the next theorem.
\begin{thm}
    For a discrete stationary memoryless source $\bfX$, additive distortion measures, and any set $\cU$ such that $|\cU| \geq |\cZ| + 1$, we have
    \begin{align}
        & \bigcup_{\bfP_{\bfU\bfY\bfZ|\bfX} \in \cP_{\mathrm{G}}(D| \bfX)} \cR_{\mathrm{G}}(\bfP_{\bfU\bfY\bfZ|\bfX}| \bfX)\notag\\
        &= \bigcup_{P_{YZ|X} \in \cP_{\mathrm{M}}(D| X)} \cR_{\mathrm{M}}(P_{YZ|X}| X).
        \label{equ: thm: general formula = i.i.d}
    \end{align}
\end{thm}
\begin{rem}
    \label{rem:CardinalBound}
    According to Theorem \ref{thm: general formula for the rate-distortion region}, the above theorem holds even if $|\cU| \geq |\cZ|$. However, we assume that $|\cU| \geq |\cZ| + 1$ in order to prove the above theorem without the help of Theorem \ref{thm: general formula for the rate-distortion region}.
\end{rem}

\begin{rem}
    \label{rem: boundary and continuity}
    The RHS of \eqref{equ: thm: general formula = i.i.d} can be written as
    \mathindent=0mm
    \begin{align}
        \{(R_1, R_2) \in \mathbb{R}_{\geq 0}^{2}: R_1 \geq R_{1}(D_1), R_1 + R_2 \geq R_{\mathrm{b}}(R_1, D_1, D_2)\},
        \label{equ: i.i.d. boundary}
    \end{align}
    \mathindent=7mm
    where $R_{1}(D_1)$ is the rate-distortion function and $R_{\mathrm{b}}(R_1, D_1, D_2)$ gives the boundary for a given $R_1$, which are defined as (see, e.g., \cite[Corollary 1]{481803}, \cite[(22)]{zhou2017second})
    \begin{align*}
        R_{1}(D_1) &\teq \min_{P_{Y|X} \in \cP_{\cY|\cX}: \mathrm{E}[d_{1}(X,Y)] \leq D_1} I(X; Y),\\
        R_{\mathrm{b}}(R_1, D_1, D_2) &\teq \min_{\substack{P_{YZ|X} \in \cP_{\cY\cZ|\cX}:\\  \mathrm{E}[d_{1}(X,Y)] \leq D_1, \mathrm{E}[d_{2}(X,Z)] \leq D_2, \\ I(X; Y) \leq R_1}} I(X; Y, Z).
    \end{align*}
    We note that $R_{1}(D_1)$ and $R_{\mathrm{b}}(R_1, D_1, D_2)$ are convex and continuous functions of a triple $(R_1, D_1, D_2)$ (see, e.g., \cite[Remark 5.2.1]{hanspringerinformation} and \cite[Lemma 4]{kostina2017successive}).
\end{rem}

We will prove the theorem by two parts separately.
\begin{proof}[{Proof: The left-hand side (LHS) of \eqref{equ: thm: general formula = i.i.d} $\subseteq$ The RHS of \eqref{equ: thm: general formula = i.i.d}}]
    We have
    \begin{align}
        & \bigcup_{\bfP_{\bfU\bfY\bfZ|\bfX} \in \cP_{\mathrm{G}}(D| \bfX)} \cR_{\mathrm{G}}(\bfP_{\bfU\bfY\bfZ|\bfX}| \bfX)\notag\\
        &\subseteq \bigcup_{\bfP_{\bfU\bfY\bfZ|\bfX} \in \cP_{\mathrm{G}}(D| \bfX)} \{(R_1,R_2): R_1\geq \overline{I}(\bfX;\bfU, \bfY),\notag\\
        &\quad R_1+R_2 \geq \overline{I}(\bfX;\bfU, \bfY) + \overline{I}(\bfX;\bfZ|\bfU, \bfY)\}\notag\\
        &\overset{\mathrm{(a)}}{\subseteq} \bigcup_{\bfP_{\bfU\bfY\bfZ|\bfX} \in \cP_{\mathrm{G}}(D| \bfX)} \{(R_1,R_2): R_1\geq \overline{I}(\bfX;\bfU, \bfY),\notag\\
        &\quad R_1+R_2 \geq \overline{I}(\bfX;\bfU,\bfY,\bfZ)\}\notag\\
        &\overset{\mathrm{(b)}}{\subseteq} \bigcup_{\substack{\bfP_{\bfY\bfZ|\bfX}:\notag\\ \overline{D}_{1}(\bfX, \bfY) \leq D_1, \overline{D}_{2}(\bfX, \bfZ) \leq D_2}}\{(R_1,R_2): R_1\geq \overline{I}(\bfX;\bfY),\notag\\
        &\quad R_1+R_2 \geq \overline{I}(\bfX;\bfY,\bfZ)\},
        \label{equ: i.i.d. starting inequality}
    \end{align}
    where (a) and (b) respectively come from the fact that
    \begin{align*}
        \overline{I}(\bfX;\bfU, \bfY) + \overline{I}(\bfX;\bfZ|\bfU, \bfY) &\geq \overline{I}(\bfX; \bfU, \bfY, \bfZ),\\
        \overline{I}(\bfX; \bfU, \bfY) &\geq \overline{I}(\bfX; \bfY).
    \end{align*}

    For a sequence $(\bfX, \bfY, \bfZ) = \{(X^n, Y^n, Z^n)\}$ of RVs induced by $\bfP_{\bfY\bfZ|\bfX}$ and a given source $\bfX$, let $(\bfX, \bar \bfY, \bar \bfZ) = \{(X^n, \bar Y^n, \bar Z^n)\}$ be a sequence of RVs such that $(X_1,\bar{Y}_{1},\bar{Z}_{1}), (X_2,\bar{Y}_{2},\bar{Z}_{2}), \cdots, (X_n,\bar{Y}_{n},\bar{Z}_{n})$ are independent and 
    \begin{align*}
        P_{X_{i}\bar{Y}_{i}\bar{Z}_{i}}(x, y, z) = P_{X_{i} Y_{i} Z_{i}}(x, y, z),
    \end{align*}
    where $P_{X_{i} Y_{i} Z_{i}}(x, y, z)$ is the $i$-th marginal distribution of $(X^n, Y^n, Z^n)$. Then, according to \cite[Lemma 5.8.1 and 5.8.2]{hanspringerinformation}, we have
    $\overline{D}_{1}(\bfX, \bfY) \geq \overline{D}_{1}(\bfX, \bar{\bfY})$, $\overline{D}_{2}(\bfX, \bfZ) \geq \overline{D}_{2}(\bfX, \bar{\bfZ})$, $\overline{I}(\bfX;\bfY) \geq \overline{I}(\bfX;\bar{\bfY})$, and $\overline{I}(\bfX;\bfY,\bfZ) \geq \overline{I}(\bfX;\bar{\bfY},\bar{\bfZ})$. Thus, by introducing the set $\cI$ of probability distributions for independent RVs as
    \begin{align*}
        \cI &\teq \left\{\bfP_{\bfY\bfZ|\bfX} = \{P_{Y^nZ^n|X^n}\}: P_{Y^nZ^n|X^n} = \prod_{i = 1}^{n} P_{Y_{i}Z_{i}|X_{i}},\right.\\
        &\quad \left. \exists P_{Y_{i}Z_{i}|X_{i}} \in \cP_{\cY\cZ|\cX}, i \in [1: n] \right\},
    \end{align*}
    we have
    \begin{align}
        &\bigcup_{\substack{\bfP_{\bfY \bfZ| \bfX}:\\ \overline{D}_{1}(\bfX, \bfY) \leq D_1, \overline{D}_{2}(\bfX, \bfZ) \leq D_2}}\{(R_1,R_2): R_1\geq \overline{I}(\bfX;\bfY),\notag\\
        &\quad R_1+R_2 \geq \overline{I}(\bfX;\bfY,\bfZ)\}\notag\\
        &\subseteq \bigcup_{\substack{\bfP_{\bfY \bfZ| \bfX} \in \cI:\\ \overline{D}_{1}(\bfX, \bfY) \leq D_1, \overline{D}_{2}(\bfX, \bfZ) \leq D_2}}\{(R_1,R_2): R_1\geq \overline{I}(\bfX;\bfY),\notag\\
        &\quad R_1+R_2 \geq \overline{I}(\bfX;\bfY,\bfZ)\}.
        \label{equ: i.i.d. general in cl of memoryless}
    \end{align}

    On the other hand, in the same way as \cite[p.372]{hanspringerinformation}, for any $\delta > 0$, $\gamma > 0$, and any $\bfP_{\bfY\bfZ|\bfX} \in \cI$ such that $\overline{D}_{1}(\bfX, \bfY)\leq D_1$ and $\overline{D}_{2}(\bfX, \bfZ)\leq D_2$, there exists $P_{YZ|X} \in \cP_{\cY\cZ|\cX}$ such that
    \begin{align*}
        \overline{I}(\bfX;\bfY)&\geq I(X;Y) - \delta,\\
        \overline{I}(\bfX;\bfY,\bfZ)&\geq I(X;Y,Z) - \delta,
    \end{align*}
    and
    \begin{align*}
        D_1 &\geq {\rm E}[d_{1}(X,Y)] - \gamma,\\
        D_2 &\geq {\rm E}[d_{2}(X,Z)] - \gamma.
    \end{align*}
    Thus, we have
    \begin{align}
        &\bigcup_{\substack{\bfP_{\bfY\bfZ|\bfX} \in \cI:\\ \overline{D}_{1}(\bfX, \bfY) \leq D_1, \overline{D}_{2}(\bfX, \bfZ) \leq D_2}}\{(R_1,R_2): R_1\geq \overline{I}(\bfX;\bfY),\notag\\
        &\quad R_1+R_2 \geq \overline{I}(\bfX;\bfY,\bfZ)\}\notag\\
        &\subseteq \bigcap_{\gamma>0}\bigcap_{\delta>0}\bigcup_{P_{YZ|X} \in \cP_{\mathrm{M}}(D + \gamma| X)}\{(R_1,R_2):\notag\\
        &\quad R_1\geq I(X;Y)-\delta, R_1+R_2 \geq I(X;Y,Z)-\delta\}\notag\\
        &= \bigcap_{\gamma>0} \bigcap_{\delta>0} \{(R_1, R_2): R_1 + \delta \geq R_{1}(D_1 + \gamma),\notag\\
        &\quad R_1 + R_2 + \delta \geq R_{\mathrm{b}}(R_1 + \delta, D_1 + \gamma, D_2 + \gamma)\},
        \label{equ: i.i.d. cl of memoryless in cl of boundary}
    \end{align}
    where the last equality comes from \eqref{equ: i.i.d. boundary}.

    Since $R_{1}(D_1)$ and $R_{\mathrm{b}}(R_1, D_1, D_2)$ are convex and continuous functions of a triple $(R_1, D_1, D_2)$ (see Remark \ref{rem: boundary and continuity}), it holds that for any $\epsilon > 0$, there exist sufficiently small $\gamma_{\epsilon} > 0$ and $\delta_{\epsilon} > 0$ such that
    \begin{align}
        &\bigcap_{\delta > 0} \bigcap_{\gamma > 0} \{(R_1, R_2): R_1 + \delta \geq R_{1}(D_1 + \gamma),\notag\\
        &\quad R_1 + R_2 + \delta \geq R_{\mathrm{b}}(R_1 + \delta, D_1 + \gamma, D_2 + \gamma)\}\notag\\
        &\subseteq \bigcap_{\delta_{\epsilon} > \delta > 0} \bigcap_{\gamma_{\epsilon}> \gamma > 0} \{(R_1, R_2): R_1 + \delta \geq R_{1}(D_1 + \gamma),\notag\\
        &\quad R_1 + R_2 + \delta \geq R_{\mathrm{b}}(R_1 + \delta, D_1 + \gamma, D_2 + \gamma)\}\notag\\
        &\subseteq \bigcap_{\delta_{\epsilon} > \delta > 0} \bigcap_{\gamma_{\epsilon}> \gamma > 0} \{(R_1, R_2): R_1 \geq R_{1}(D_1) - \delta - \epsilon,\notag\\
        &\quad R_1 + R_2 \geq R_{\mathrm{b}}(R_1, D_1, D_2) - \delta - \epsilon\}\notag\\
        &= \bigcap_{\delta_{\epsilon} > \delta > 0} \{(R_1, R_2): R_1\geq R_{1}(D_1) - \delta - \epsilon,\notag\\
        &\quad R_1+ R_2 \geq R_{\mathrm{b}}(R_1, D_1, D_2) - \delta - \epsilon\}.
        \label{equ: i.i.d. continuity of boundary points}
    \end{align}
    By combining \eqref{equ: i.i.d. starting inequality}, \eqref{equ: i.i.d. general in cl of memoryless}, \eqref{equ: i.i.d. cl of memoryless in cl of boundary}, and \eqref{equ: i.i.d. continuity of boundary points}, we have
    \begin{align*}
        & \bigcup_{\bfP_{\bfU\bfY\bfZ|\bfX} \in \cP_{\mathrm{G}}(D| \bfX)} \cR_{\mathrm{G}}(\bfP_{\bfU\bfY\bfZ|\bfX}| \bfX)\notag\\
        &\subseteq \bigcap_{\epsilon > 0} \bigcap_{\delta_{\epsilon} > \delta > 0} \{(R_1, R_2): R_1\geq R_{1}(D_1) - \delta - \epsilon,\notag\\
        &\quad R_1 + R_2 \geq R_{\mathrm{b}}(R_1, D_1, D_2) - \delta - \epsilon\}\notag\\
        &= \{(R_1, R_2): R_1\geq R_{1}(D_1), R_1 + R_2 \geq R_{\mathrm{b}}(R_1, D_1, D_2)\}.
    \end{align*}
    According to Remark \ref{rem: boundary and continuity}, this completes the proof.
\end{proof}

\begin{proof}[{Proof: The LHS of \eqref{equ: thm: general formula = i.i.d} $\supseteq$ The RHS of \eqref{equ: thm: general formula = i.i.d}}]    
    Since $|\cU| \geq |\cZ| + 1$, there exists an injective function $\mathrm{id}: \cZ \ra \cU$. For this function, let $\cU_{\mathrm{id}} \subseteq \cU$ be the image of $\mathrm{id}$ and $\mathrm{id}^{-1}: \cU_{\mathrm{id}} \ra \cZ$ be the inverse function of $\mathrm{id}$ on $\cU_{\mathrm{id}}$. Let $u^* \in \cU$ be a symbol such that $u^* \notin \cU_{\mathrm{id}}$.
    
    For any $P_{YZ|X} \in \cP_{\mathrm{M}}(D| X)$ and any $\alpha \in [0, 1]$, we define $P_{UYZ|X} \in \cP_{\cU\cY\cZ| \cX}$ as
    \mathindent=0mm
    \begin{align*}
        P_{UYZ|X}(u, y, z| x) \teq
        \begin{cases}
            \alpha P_{YZ|X}(y, z| x) & \mbox{ if } u = \mathrm{id}(z),\\
            (1 - \alpha) P_{YZ|X}(y, z| x) & \mbox{ if } u = u^*,\\
            0 & \mbox{ otherwise}.
        \end{cases}
    \end{align*}
    \mathindent=7mm
    Since 
    \begin{align*}
        P_{XUY}(x, u, y) &= 
        \begin{cases}
            \alpha P_{XYZ}(x, y, \mathrm{id}^{-1}(u)) & \mbox{ if } u \in \cU_{\mathrm{id}},\\
            (1 - \alpha) P_{XY}(x, y) & \mbox{ if } u = u^*,\\
            0 & \mbox{ otherwise},
        \end{cases}\\
        P_{X|UY}(x| u, y) &=
        \begin{cases}
            P_{X|YZ}(x| y, \mathrm{id}^{-1}(u)) & \mbox{ if } u \in \cU_{\mathrm{id}},\\
            P_{X|Y}(x| y) & \mbox{ if } u = u^*,\\
            0 & \mbox{ otherwise},
        \end{cases}
    \end{align*}
    we have
    \begin{align*}
        I(X; U, Y) 
        &= H(X) - \alpha H(X|Y, Z) - (1 - \alpha) H(X|Y)\\
        &= \alpha I(X; Y, Z) + (1 - \alpha) I(X; Y).
    \end{align*}
    Furthermore, since
    \begin{align*}
        P_{X|UYZ}(x|u, y, z) &=
        \begin{cases}
            P_{X|YZ}(x|y, z) & \mbox{ if } u = \mathrm{id}(z),\\
            P_{X|YZ}(x|y, z) & \mbox{ if } u = u^*,\\
            0 & \mbox{ otherwise},
        \end{cases}
    \end{align*}
    we have
    \mathindent=0mm
    \begin{align*}
        I(X; Z| U, Y) &= H(X|U, Y) - H(X|U, Y, Z)\\
        &= \alpha H(X|Y, Z) + (1 - \alpha) H(X|Y) - H(X|Y, Z)\\
        &= (1 - \alpha) I(X; Z| Y).
    \end{align*}
    \mathindent=7mm
    
    Now by defining $\bfP_{\bfU\bfY\bfZ|\bfX}$ as
    \begin{align*}
        P_{U^nY^nZ^n|X^n}(u^n, y^n, z^n| x^n) &= \prod_{i = 1}^{n} P_{UYZ|X}(u_i, y_i, z_i| x_i),
    \end{align*}
    $(\bfX, \bfU,\bfY,\bfZ)$ becomes a sequence of independent copies of RVs $(X, U, Y, Z)$. Thus, we have
    \mathindent=0mm    
    \begin{align*}
        \overline{I}(\bfX; \bfU, \bfY) &= I(X; U, Y) = \alpha I(X; Y, Z) + (1 - \alpha) I(X; Y),\\
        \overline{I}(\bfX;\bfZ|\bfU, \bfY) &= I(X; Z| U, Y) = (1 - \alpha) I(X; Z| Y),\\
        \overline{D}_{1}(\bfX, \bfY) &= \mathrm{E}[d_{1}(X, Y)] \leq D_1,\\
        \overline{D}_{2}(\bfX, \bfZ) &= \mathrm{E}[d_{2}(X, Z)] \leq D_2,
    \end{align*}
    \mathindent=7mm
    where we use the fact that for i.i.d.\ RVs $\{A_{i}\}_{i = 1}^{\infty}$ taking values in a finite set, $\plimsup_{n \ra \infty} \frac{1}{n} \sum_{i = 1}^{n} A_{i} = \mathrm{E}[A_{1}]$. Hence, by noting that $\bfP_{\bfU\bfY\bfZ|\bfX} \in \cP_{\mathrm{G}}(D| \bfX)$, we have for any $\alpha \in [0, 1]$,
    \begin{align}
        &(\alpha I(X; Y, Z) + (1 - \alpha) I(X; Y), (1 - \alpha) I(X; Z| Y))\notag\\
        &\quad \in \bigcup_{\bfP_{\bfU\bfY\bfZ|\bfX} \in \cP_{\mathrm{G}}(D| \bfX)} \cR_{\mathrm{G}}(\bfP_{\bfU\bfY\bfZ|\bfX}| \bfX).
        \label{equ: memoryless alpha rate in RG}
    \end{align}
    This implies that for any $(R_1, R_2) \in \cR_{\mathrm{M}}(P_{YZ|X}| X)$, it holds that $(R_1, R_2) \in \bigcup_{\bfP_{\bfU\bfY\bfZ|\bfX} \in \cP_{\mathrm{G}}(D| \bfX)} \cR_{\mathrm{G}}(\bfP_{\bfU\bfY\bfZ|\bfX}| \bfX)$. This is because for any $(R_1, R_2) \in \cR_{\mathrm{M}}(P_{YZ|X}| X)$ such that $I(X;Y,Z) \geq R_1$, there exists $\alpha \in [0, 1]$ such that
    \begin{align*}
        R_1 &= \alpha I(X;Z|Y) + I(X;Y)\\
        &= \alpha I(X;Y,Z)+(1-\alpha)I(X;Y).
    \end{align*}
    Then, we have
    \begin{align}
        R_2&\geq I(X;Y,Z) - R_1\notag\\
        &= I(X;Y,Z) - \alpha I(X;Y,Z) - (1-\alpha)I(X;Y)\notag\\
        &= (1-\alpha)I(X;Z|Y).
        \label{equ: time-sharing argument}
    \end{align}
    According to \eqref{equ: memoryless alpha rate in RG}, such pair $(R_1, R_2)$ is included in the region $\bigcup_{\bfP_{\bfU\bfY\bfZ|\bfX} \in \cP_{\mathrm{G}}(D| \bfX)} \cR_{\mathrm{G}}(\bfP_{\bfU\bfY\bfZ|\bfX}| \bfX)$. On the other hand, for any $(R_1, R_2) \in \cR_{\mathrm{M}}(P_{YZ|X}| X)$ such that $R_1 > I(X;Y,Z)$, we have $R_2\geq 0$. Since it holds that $(I(X;Y,Z),0) \in \bigcup_{\bfP_{\bfU\bfY\bfZ|\bfX} \in \cP_{\mathrm{G}}(D| \bfX)} \cR_{\mathrm{G}}(\bfP_{\bfU\bfY\bfZ|\bfX}| \bfX)$ due to \eqref{equ: memoryless alpha rate in RG}, such pair $(R_1, R_2)$ is also included in the region $\bigcup_{\bfP_{\bfU\bfY\bfZ|\bfX} \in \cP_{\mathrm{G}}(D| \bfX)} \cR_{\mathrm{G}}(\bfP_{\bfU\bfY\bfZ|\bfX}| \bfX)$.

    Therefore, we have
    \begin{align*}
        \bigcup_{\bfP_{\bfU\bfY\bfZ|\bfX} \in \cP_{\mathrm{G}}(D| \bfX)} \cR_{\mathrm{G}}(\bfP_{\bfU\bfY\bfZ|\bfX}| \bfX)
        \supseteq \cR_{\mathrm{M}}(P_{YZ|X}| X).
    \end{align*}
    Since this holds for any $P_{YZ|X} \in \cP_{\mathrm{M}}(D| X)$, this completes the proof.
\end{proof}

\begin{rem}
    Unlike the rate-distortion region by Rimoldi, our region includes a sequence  $\bfU$ of auxiliary RVs. This comes from the fact that the time-sharing argument as in \eqref{equ: time-sharing argument} cannot be employed because it holds that in general,
    \begin{align*}
        \overline{I}(\bfX; \bfY, \bfZ) \neq \overline{I}(\bfX; \bfY) + \overline{I}(\bfX; \bfZ| \bfY).
    \end{align*}  
\end{rem}

\subsection{Mixed Sources}
In this section, we give the rate-distortion region for mixed sources.

The mixed source $\bfX$ is defined by $\bfX_{1} = \{X_{1}^{n}\}_{n=1}^{\infty}$ and $\bfX_{2} = \{X_{2}^{n}\}_{n=1}^{\infty}$ as
\begin{align*}
    P_{X^n}(x^n) = \alpha_{1} P_{X_{1}^n}(x^n) + \alpha_{2} P_{X_{2}^n}(x^n),
\end{align*}
where $\alpha_1, \alpha_2 \in [0, 1]$ and $\alpha_1 + \alpha_2 = 1$.

The next lemma shows a fundamental property of the information spectrum of mixed sources.
\begin{lem}
    \label{lem: mixed source is max mutual information}
    For sequences of RVs $(\bfX_{1}, \bfY_{1}, \bfZ_{1})$ and $(\bfX_{2}, \bfY_{2}, \bfZ_{2})$, let $(\bfX, \bfY, \bfZ)$ be defined by
    \begin{align*}
        P_{X^nY^nZ^n}(x^n, y^n, z^n)
        &=\alpha_{1} P_{X_{1}^nY_{1}^nZ_{1}^n}(x^n, y^n, z^n)\\
        &\quad + \alpha_{2} P_{X_{2}^nY_{2}^nZ_{2}^n}(x^n, y^n, z^n).
    \end{align*}
    Then, we have
    \begin{align*}
        \overline{I}(\bfX; \bfY) &= \max\{\overline{I}(\bfX_{1}; \bfY_{1}), \overline{I}(\bfX_{2}; \bfY_{2})\},\\
        \overline{I}(\bfX; \bfY| \bfZ) &= \max\{\overline{I}(\bfX_{1}; \bfY_{1}|\bfZ_{1}), \overline{I}(\bfX_{2}; \bfY_{2}|\bfZ_{2})\}.
    \end{align*}
\end{lem}
\begin{proof}
    Since this lemma can be proved in the same way as \cite[Lemma 7.9.1]{hanspringerinformation} by using \cite[Lemma 1.4.2]{hanspringerinformation}, we omit the details.
\end{proof}

The next theorem shows that the rate-distortion region for a mixed source is the intersection of those of two sources.
\begin{thm}
    For a mixed source $\bfX$ defined by $\bfX_{1}$ and $\bfX_{2}$, and any real numbers $D_1, D_2 \geq 0$, we have
    \begin{align*}
        \cR(D| \bfX) = \cR(D| \bfX_{1}) \cap \cR(D| \bfX_{2}).
    \end{align*}
\end{thm}

\begin{proof}
    For $\bfP_{\bfU_{1}\bfY_{1}\bfZ_{1}| \bfX_{1}}$ and $\bfP_{\bfU_{2}\bfY_{2}\bfZ_{2}| \bfX_{2}}$, we define a mixture $\bfP_{\tilde\bfU\tilde\bfY\tilde\bfZ| \bfX}$ of these two components by
    \begin{align*}
        &P_{\tilde U^n\tilde Y^n\tilde Z^n| X^n}(u^n, y^n, z^n| x^n)\\
        &\quad = \frac{\alpha_{1} P_{X_{1}^n}(x^n) P_{U_{1}^nY_{1}^nZ_{1}^n| X_{1}^n}(u^n, y^n, z^n| x^n)}{P_{X^n}(x^n)}\\
        &\qquad + \frac{\alpha_{2} P_{X_{2}^n}(x^n) P_{U_{2}^nY_{2}^nZ_{2}^n| X_{2}^n}(u^n, y^n, z^n| x^n)}{P_{X^n}(x^n)}.
    \end{align*}
    In order to prove the theorem, we give an equivalent expression of the rate-distortion region using $\bfP_{\bfU_{1}\bfY_{1}\bfZ_{1}| \bfX_{1}}$, $\bfP_{\bfU_{2}\bfY_{2}\bfZ_{2}| \bfX_{2}}$, and the mixture $\bfP_{\tilde\bfU\tilde\bfY\tilde\bfZ| \bfX}$.

    When $\bfP_{\bfU_{1}\bfY_{1}\bfZ_{1}| \bfX_{1}} = \bfP_{\bfU_{2}\bfY_{2}\bfZ_{2}| \bfX_{2}} = \bfP_{\bfU\bfY\bfZ| \bfX}$, we have $\bfP_{\tilde\bfU\tilde\bfY\tilde\bfZ| \bfX} = \bfP_{\bfU\bfY\bfZ| \bfX}$ by the definition. This implies that for any $\bfP_{\bfU\bfY\bfZ| \bfX}$, there exist $\bfP_{\bfU_{1}\bfY_{1}\bfZ_{1}| \bfX_{1}}$ and $\bfP_{\bfU_{2}\bfY_{2}\bfZ_{2}| \bfX_{2}}$ such that $\bfP_{\bfU\bfY\bfZ| \bfX} = \bfP_{\tilde\bfU\tilde\bfY\tilde\bfZ| \bfX}$. On the other hand, for any $\bfP_{\bfU_{1}\bfY_{1}\bfZ_{1}| \bfX_{1}}$ and $\bfP_{\bfU_{2}\bfY_{2}\bfZ_{2}| \bfX_{2}}$, there trivially exists $\bfP_{\bfU\bfY\bfZ| \bfX}$ such that $\bfP_{\bfU\bfY\bfZ| \bfX} = \bfP_{\tilde\bfU\tilde\bfY\tilde\bfZ| \bfX}$. Thus, we have an equivalent expression:
    \begin{align*}
        &\bigcup_{\bfP_{\bfU\bfY\bfZ|\bfX} \in \cP_{\mathrm{G}}(D| \bfX)} \cR_{\mathrm{G}}(\bfP_{\bfU\bfY\bfZ|\bfX}| \bfX)\\
        &= \bigcup_{\substack{\bfP_{\bfU_{1}\bfY_{1}\bfZ_{1}| \bfX_{1}}, \bfP_{\bfU_{2}\bfY_{2}\bfZ_{2}| \bfX_{2}}:\\ \overline{D}_{1}(\bfX, \tilde \bfY) \leq D_1, \overline{D}_{2}(\bfX, \tilde \bfZ) \leq D_2}} \cR_{\mathrm{G}}(\bfP_{\tilde \bfU \tilde \bfY \tilde \bfZ|\bfX}| \bfX).
    \end{align*}

    Thus, according to Theorem \ref{thm: general formula for the rate-distortion region}, we have
    \begin{align*}
        &\cR(D| \bfX)\\
        &=\bigcup_{\substack{\bfP_{\bfU_{1}\bfY_{1}\bfZ_{1}| \bfX_{1}}, \bfP_{\bfU_{2}\bfY_{2}\bfZ_{2}| \bfX_{2}}:\\ \overline{D}_{1}(\bfX, \tilde \bfY) \leq D_1, \overline{D}_{2}(\bfX, \tilde \bfZ) \leq D_2}} \cR_{\mathrm{G}}(\bfP_{\tilde \bfU \tilde \bfY \tilde \bfZ|\bfX}| \bfX)\\
        &\eqo{(a)} \bigcup_{\substack{\bfP_{\bfU_{1}\bfY_{1}\bfZ_{1}| \bfX_{1}}, \bfP_{\bfU_{2}\bfY_{2}\bfZ_{2}| \bfX_{2}}:\\ \max\{\overline{D}_{1}(\bfX_{1}, \bfY_{1}), \overline{D}_{1}(\bfX_{2}, \bfY_{2})\} \leq D_1,\\ \max\{\overline{D}_{2}(\bfX_{1}, \bfZ_{1}), \overline{D}_{2}(\bfX_{2}, \bfZ_{2})\} \leq D_2}} \{(R_1,R_2):\\
        &\quad R_1\geq \max\{\overline{I}(\bfX_{1}; \bfU_{1}, \bfY_{1}), \overline{I}(\bfX_{2}; \bfU_{2}, \bfY_{2})\}\\
        &\quad R_2\geq \max\{\overline{I}(\bfX_{1}; \bfZ_{1}| \bfU_{1}, \bfY_{1}), \overline{I}(\bfX_{2}; \bfZ_{2}| \bfU_{2}, \bfY_{2})\}\}\\
        &= \bigcup_{\bfP_{\bfU_{1}\bfY_{1}\bfZ_{1}| \bfX_{1}} \in \cP_{\mathrm{G}}(D| \bfX_{1})} \bigcup_{\bfP_{\bfU_{2}\bfY_{2}\bfZ_{2}| \bfX_{2}} \in \cP_{\mathrm{G}}(D| \bfX_{2})}\\
        &\quad \cR_{\mathrm{G}}(\bfP_{\bfU_{1}\bfY_{1}\bfZ_{1}|\bfX_{1}}| \bfX_{1}) \cap \cR_{\mathrm{G}}(\bfP_{\bfU_{2}\bfY_{2}\bfZ_{2}|\bfX_{2}}| \bfX_{2})\\
        &= \cR(D|\bfX_{1}) \cap \cR(D|\bfX_{2}).
    \end{align*}
    where (a) comes from \cite[Lemma 1.4.2]{hanspringerinformation}, Lemma \ref{lem: mixed source is max mutual information} and the fact that
    \begin{align*}
        &P_{X^n\tilde U^n\tilde Y^n\tilde Z^n}(x^n, u^n, y^n, z^n)\\
        &\quad = \alpha_{1} P_{X_{1}^nU_{1}^nY_{1}^nZ_{1}^n}(x^n, u^n, y^n, z^n)\\
        &\qquad + \alpha_{2} P_{X_{2}^nU_{2}^nY_{2}^nZ_{2}^n}(x^n, u^n, y^n, z^n).
    \end{align*}
    This completes the proof.
\end{proof}

\section{Conclusion}
\label{sec: conclusion}
In this paper, we have dealt with the successive refinement problem. We gave inner and outer bounds using the smooth max R\'enyi divergence on the set of pairs of numbers of codewords. These bounds are obtained by extended versions of our previous covering lemma and converse bound. By using these bounds, we also gave a general formula using the spectral sup-mutual information rate for the rate-distortion region. Further, we showed some special cases of our rate-distortion region for discrete stationary memoryless sources and mixed sources.

\section*{Acknowledgment}
The authors would like to thank the anonymous reviewers for their valuable comments. This work was supported in part by JSPS KAKENHI Grant Number 15K15935.


\begin{thebibliography}{10}

\bibitem{matsuta2015nab}
T.~Matsuta and T.~Uyematsu, ``Non-asymptotic bounds on numbers of codewords for
  the successive refinement problem,'' Proc. 38th Symp. on Inf. Theory and its
  Apps. (SITA2015), pp.43--48, Nov.\ 2015.

\bibitem{matsuta2016general}
T.~Matsuta and T.~Uyematsu, ``A general formula of the achievable rate region
  for the successive refinement problem,'' Proc. IEICE Society Conference,
  p.37, Sep.\ 2016.

\bibitem{koshelev1981estimation}
V.~N.~Koshelev, ``Estimation of mean error for a discrete
  successive-approximation scheme,'' Problemy Peredachi Informatsii, vol.17,
  no.3, pp.20--33, 1981.

\bibitem{koshelev1980hierarchical}
V.~N.~Koshelev, ``Hierarchical coding of discrete sources,'' Problemy Peredachi
  Informatsii, vol.16, no.3, pp.31--49, 1980.

\bibitem{equitz1991successive}
W.~Equitz and T.~Cover, ``Successive refinement of information,'' IEEE Trans.
  Inf. Theory, vol.37, no.2, pp.269--275, Mar.\ 1991.

\bibitem{272493}
B.~Rimoldi, ``Successive refinement of information: characterization of the
  achievable rates,'' IEEE Trans. Inf. Theory, vol.40, no.1, pp.253--259, Jan.\
  1994.

\bibitem{490549}
H.~Yamamoto, ``Source coding theory for a triangular communication system,''
  IEEE Trans. Inf. Theory, vol.42, no.3, pp.848--853, May\ 1996.

\bibitem{782109}
M.~Effros, ``Distortion-rate bounds for fixed- and variable-rate
  multiresolution source codes,'' IEEE Trans. Inf. Theory, vol.45, no.6,
  pp.1887--1910, Sep.\ 1999.

\bibitem{7445859}
A.~No, A.~Ingber, and T.~Weissman, ``Strong successive refinability and
  rate-distortion-complexity tradeoff,'' IEEE Trans. Inf. Theory, vol.62, no.6,
  pp.3618--3635, June\ 2016.

\bibitem{zhou2017second}
L.~Zhou, V.Y.F. Tan, and M.~Motani, ``Second-order and moderate deviations
  asymptotics for successive refinement,'' IEEE Trans. Inf. Theory, vol.63,
  no.5, pp.2896--2921, May\ 2017.

\bibitem{6691300}
N.A. Warsi, ``One-shot bounds for various information theoretic problems using
  smooth min and max {R\'enyi} divergences,'' Proc. 2013 IEEE Inf. Theory
  Workshop, pp.1--5, Sep.\ 2013.

\bibitem{uyematsu2014ITWrevisiting}
T.~Uyematsu and T.~Matsuta, ``Revisiting the rate-distortion theory using
  smooth max {R\'enyi} divergence,'' Proc. 2014 IEEE Inf. Theory Workshop,
  pp.202--206, Nov.\ 2014.

\bibitem{matsuta2016new}
T.~Matsuta and T.~Uyematsu, ``New non-asymptotic bounds on numbers of codewords
  for the fixed-length lossy compression,'' IEICE Trans. Fundamentals,
  vol.E99-A, no.12, pp.2116--2129, Dec.\ 2016.

\bibitem{hanspringerinformation}
T.~S.~Han, {Information-Spectrum Methods in Information Theory}, Springer,
  2003.

\bibitem{kostina2012fixed}
V.~Kostina and S.~Verd{\'u}, ``Fixed-length lossy compression in the finite
  blocklength regime,'' IEEE Trans. Inf. Theory, vol.58, no.6, pp.3309--3338,
  June\ 2012.

\bibitem{cover2006eit}
T.~M.~Cover and J.~A.~Thomas, {Elements of Information Theory}, 2nd~ed., Wiley,
  New York, 2006.

\bibitem{kostina2017rate-distortion}
V.~Kostina and E.~Tuncel, ``The rate-distortion function for successive
  refinement of abstract sources,'' Proc. IEEE Int. Symp. on Inf. Theory,
  pp.1923--1927, June\ 2017.

\bibitem{verdu2015alpha}
S.~Verd\'u, ``$\alpha$-mutual information,'' 2015 Inf. Theory and Apps.
  Workshop, pp.1--6, Feb.\ 2015.

\bibitem{481803}
A.~Kanlis and P.~Narayan, ``Error exponents for successive refinement by
  partitioning,'' IEEE Trans. Inf. Theory, vol.42, no.1, pp.275--282, Jan.\
  1996.

\bibitem{kostina2017successive}
V.~Kostina and E.~Tuncel, ``Successive refinement of abstract sources,'' arXiv
  preprint arXiv:1707.09567, July\ 2017.

\end{thebibliography}
\end{document}